\documentclass[10pt,twocolumn,english]{IEEEtran}

\usepackage{amsmath,amsfonts,epsfig,algorithm,algorithmic,subfig,url}
\usepackage[amsmath,thmmarks]{ntheorem}

\def\RN{ \mathbb{R} }                               
\def\ZN{ \mathbb{Z} }                               
\newcommand{\EE}{{\mathbb E}}                       
\newcommand{\bth}{\boldsymbol{\theta}}              

\newcommand{\bdelta}{{\boldsymbol \delta}}          

\newcommand{\ba}{\boldsymbol{a}}                    
\newcommand{\bA}{\boldsymbol{A}}                    
\newcommand{\bB}{\boldsymbol{B}}                    
\newcommand{\bC}{\boldsymbol{C}}                    
\newcommand{\bD}{\boldsymbol{D}}                    
\newcommand{\bE}{\boldsymbol{E}}                    
\newcommand{\bzero}{\boldsymbol{0}}                 
\newcommand{\bb}{\boldsymbol{b}}                    
\newcommand{\bc}{\boldsymbol{c}}                    
\newcommand{\bd}{\boldsymbol{d}}                    
\newcommand{\be}{\boldsymbol{e}}                    
\newcommand{\bg}{\boldsymbol{g}}                    
\newcommand{\bI}{\boldsymbol{I}}                    
\newcommand{\bQ}{\boldsymbol{Q}}                    
\newcommand{\zero}{\boldsymbol{0}}

\newcommand{\HH}{{\mathcal H}}
\newcommand{\XX}{{\mathcal X}}
\newcommand{\AAA}{{\mathcal A}}

\newcommand{\pd}[2]{\frac{\partial #1}{\partial #2}}
\newcommand{\Ra}[1]{{{\mathcal R}\! \left( #1 \right) }}
\newcommand{\Nu}[1]{{{\mathcal N}\! \left( #1 \right) }}

\DeclareMathOperator{\sinc}{sinc}

\newcommand{\eg}{{\em e.g., }}

\newtheorem{theorem}{Theorem}

\newtheorem{proposition}{Proposition}

\begin{document}

\title{Xampling at the Rate of Innovation}

\author{Tomer Michaeli and Yonina C. Eldar,~\IEEEmembership{Senior~Member,~IEEE} %
\thanks{This work was accepted for publication in IEEE Transactions on Signal Processing.%
} 
\thanks{This work was supported in part by the Israel Science Foundation under
Grant no.~170/10 and by a Google Research Award.%
} %
\thanks{The authors are with the department of electrical engineering, Technion--Israel
Institute of Technology, Haifa, Israel 32000 (phone: +972-4-8294700, +972-4-8293256, fax: +972-4-8295757, e-mail: tomermic@tx.technion.ac.il, yonina@ee.technion.ac.il).%
}}
\maketitle
\begin{abstract}
We address the problem of recovering signals from samples taken at their rate of innovation. Our only assumption is that the sampling system is such that the parameters defining the signal can be stably determined from the samples, a condition that lies at the heart of every sampling theorem. Consequently, our analysis subsumes previously studied nonlinear acquisition devices and nonlinear signal classes. In particular, we do not restrict attention to memoryless nonlinear distortions or to union-of-subspace models. This allows treatment of various finite-rate-of-innovation (FRI) signals that were not previously studied, including, for example, continuous phase modulation transmissions. Our strategy relies on minimizing the error between the measured samples and those corresponding to our signal estimate. This least-squares (LS) objective is generally non-convex and might possess many local minima. Nevertheless, we prove that under the stability hypothesis, any optimization method designed to trap a stationary point of the LS criterion necessarily converges to the true solution. We demonstrate our approach in the context of recovering pulse streams in settings that were not previously treated. Furthermore, in situations for which other algorithms are applicable, we show that our method is often preferable in terms of noise robustness.
\end{abstract}


\section{Introduction}
\label{sec:introduction}

Sampling theory is concerned with recovery of continuous-time signals from their samples. Being an under-determined problem, sampling theorems often rely on the assumption that the signal to be recovered belongs to some predefined class of functions. The ``richness'' of this class dictates a minimal sampling rate required for perfect reconstruction. For example, the well known Shannon sampling theorem \cite{shannon1949} states that any signal $x(t)$ that is $\pi/T$-bandlimited can be perfectly recovered from its pointwise uniformly-spaced samples if the sampling interval does not exceed $T$. Similarly, if $x(t)$ is known to belong to the class of spline functions with knots at $t=nT$, $n\in\ZN$, then it can be recovered from pointwise uniform samples with interval $T$ \cite{unser1993}.

Until recently, much of the sampling literature treated linear acquisition devices and linear signal priors, that is, families of signals that form subspaces of $L_2$ \cite{eldar2009}. These include shift-invariant spaces \cite{aldroubi1994}, of which the bandlimited and spline priors are special cases, and their generalizations \cite{aldroubi1996}. Reconstruction in SI spaces from nonuniform pointwise samples was treated in \cite{aldroubi01}. Recovery from linear measurements in arbitrary subspaces was studied from an abstract Hilbert space viewpoint in \cite{eldar2003,eldar2006,michaeli10}. The appeal of subspace models and linear sampling stems from the fact that they result in linear recovery algorithms that are often easy to implement. However, many real-world signal classes do not conform to the subspace model and practical samplers often introduce nonlinear distortions \cite{dvorkind2008}.

One notable line of work deviating from these settings treats \emph{nonlinear sampling of linear models}. The first contributions in this direction can be attributed to \cite{landau61,sandberg63}, which studied reconstruction of bandlimited signals from companding (namely, applying a memoryless nonlinear distortion) and subsequent bandlimiting. These results were later extended to stochastic processes \cite{masry73} and to more general spaces \cite{sandberg94}. In \cite{dvorkind2008}, the authors generalized these developments to the setting in which the linear part of the acquisition device does not necessarily match the signal prior. A simpler iterative algorithm, consisting of linear time-invariant (LTI) filtering operations, was recently developed in \cite{faktor10} for the same setting.

Another, rather parallel, deviation from the widely studied linear setting treats \emph{linear sampling of nonlinear models}. Notable in this respect is the study initiated in \cite{vetterli2002} of sampling finite rate of innovation (FRI) signals. Theses signal classes correspond to families of functions defined by a finite number $\rho$ of parameters per time unit, a quantity referred to as their rate of innovation. Much of the recent attention attracted by this field emerges from the observation that several commonly-encountered FRI signals can be perfectly recovered from samples taken at their rate of innovation. Specifically, in \cite{vetterli2002}, it was demonstrated how periodic and finite-duration streams of Diracs, nonuniform splines and piecewise polynomials can be recovered from uniformly-spaced samples taken at the rate of innovation with either a sinc or a Gaussian kernel. Extensions to certain infinite-duration signals as well as more general classes of sampling kernels appeared in \cite{dragotti2007}, though at the cost of an increase in the sampling rate beyond the rate of innovation. A family of finite-duration sampling kernels was presented in \cite{tur2011} and demonstrated to substantially improve recovery stability. A robust multichannel sampling scheme was recently proposed in \cite{gedalyahu2011}. Finally, the authors of \cite{gedalyahu2010} studied sampling of a class of semi-periodic functions at the minimal possible rate, using a filter-bank of properly chosen filters.

All the works mentioned above for linear sampling of nonlinear models focused on signals that can be represented as weighted combinations of shifted pulses. These signal classes correspond to unions of subspaces \cite{lu2008}. Another important family within the union-of-subspace category is the set of multiband signals. As shown in \cite{mishali2009}, when using point-wise samples, the minimal sampling rate required for perfect recovery of these signals is twice their Landau rate, defined as twice the length of the support in the frequency domain. A low-rate multi-coset sampling method for multi-band signals was proposed in \cite{mishali2009}. A more practical multichannel sampling system was later developed \cite{mishali2010} and implemented on a board \cite{mishali2011}. An important feature of these systems is that the low-rate samples can be used directly to perform digital processing operations, without requiring reconstruction of the analog signal or its high-rate samples as an initial step. This is the key in the recently introduced Xampling paradigm for sampling signals that lie in a union of subspaces \cite{mishali2012,mishali2011b}.

Both lines of work treating nonlinear sampling of linear models and linear sampling of nonlinear models lack the full generality required for deployment in a wide range of practical systems. In particular, common to all nonlinear sampling works is the assumption that the nonlinearity is memoryless, such as in the Wiener-Hammerstein model treated in \cite{dvorkind2008}. However, this is not the case in many real-world applications. An exception is \cite{frank96}, which treats Volterra systems, but only focuses on bandlimited signals and point-wise samples. Similarly, all nonlinear models treated in the literature correspond to unions of subspaces, with the vast majority focusing on pulse streams. These do not include, for example, FRI signals such as continuous-phase modulation (CPM) transmissions. Furthermore, even within the restricted category of pulse streams, solutions are only available for a few special cases of signal structures and sampling devices. These solutions are very unstable in certain situations \cite{ben-haim2011}. An iterative algorithm for reconstructing signals lying in unions of subspaces from linear samples was proposed in \cite{blumensath2011}. The disadvantage of this technique, though, is that it requires, in each iteration, computing the orthogonal projection of the current signal estimate onto the set of all feasible signals. For most interesting signal models, this necessitates solving a non-convex optimization problem, which does not admit a closed form solution and for which there is no guarantee that standard optimization techniques will find its solution.

In this paper, we address the problem of reconstructing arbitrary FRI signals from possibly nonlinear measurements obtained at the rate of innovation. The only assumption we make on the sampling mechanism and signal prior is that the parameters defining the signal can be stably recovered from the samples. This assumption must be made by any practical sampling theorem that attempts to recover the signal parameters, whether explicitly or implicitly. Our approach is based on minimization of the error norm between the given set of samples and those of our signal estimate. Our main result is that under the stability assumption, this least-squares (LS) criterion possesses a unique stationary point. Consequently, any optimization algorithm designed to trap a stationary point, will necessarily converge to the true parameters. In particular, we show that the steepest-descent and quasi-Newton methods can be used to recover the signal parameters.

Our approach holds several important advantages. First, it is suited to a family of problems, which supersedes those treated by existing techniques. In particular, we do not assume that the sampling mechanism is linear or that the class of feasible signals forms a union of subspaces. Second, it provides a unified framework for recovering signals from samples taken at their rate of innovation. Thus, rather than tailoring a different algorithm for every possible combination of sampling method and signal prior, we can apply the same optimization technique to recover the signal parameters. Lastly, our method directly extracts the parameters defining the signal, which are the quantities of interest in most applications, thus conveniently allowing for further digital processing. For example, the parameters can correspond to transmitted symbols in a communication setting, reflector locations in ultrasound imaging \cite{tur2011}, and more. These properties all align with the Xampling methodology \cite{mishali2011b} and even broaden it to beyond the standard linear sampling and union-of-subspace settings.

It is important to note that our approach requires that all feasible signals can be stably recovered from the samples. Thus, even if a specific signal can theoretically be stably recovered, our method is not guaranteed to succeed when there exist other feasible signals which cannot be stably reconstructed. We demonstrate this limitation in the context of a concrete example in Section~\ref{sec:stability}.

The paper is organized as follows. In Section~\ref{sec:ProblemSetting} we describe the problem setting and assumptions. In Section~\ref{sec:SamplingRate} we derive a lower bound on the minimal sampling rate required for perfect recovery with a given sampling system. Next, in Section~\ref{sec:LSrecovery}, we describe and prove the validity of a general strategy for recovering signals from samples taken at the minimal rate. Two practical iterative methods are then studied in detail in Section~\ref{sec:IterativeRecovery}. Finally, we demonstrate our approach in the context of finite-duration and periodic pulse-stream recovery in Section~\ref{sec:ApplicationChannelSounding} and in the context of CPM receivers in Section~\ref{sec:ApplicationFSK}. We show that our method can cope with sampling systems beyond those previously studied. Furthermore, we demonstrate that in time-delay settings for which other algorithms are applicable, our method is often more robust to noise.

\section{Problem Setting}
\label{sec:ProblemSetting}

We denote scalars by lowercase letters, vectors by bold lowercase letters and matrices by bold uppercase letters (\eg $a\in\RN$, $\ba\in\RN^N$ and $\bA\in\RN^{M\times N}$). The adjoint of a linear operator $S$ is denoted $S^*$ and its null space and range space are written as $\Nu{S}$ and $\Ra{S}$ respectively. If $h$ is a function from some Hilbert space $\HH_1$ to another Hilbert space $\HH_2$, then its Fr\'echet derivative at $x_0$ is a continuous linear operator $(\partial{h}/\partial{x})|_{x_0} :\HH_1\rightarrow\HH_2$ such that
\begin{equation}
\lim_{\bdelta \rightarrow 0} \frac{ \left\| h(x_0 + \delta) - h(x_0) - \left.\pd{h}{x}\right|_{x_0} \delta \right\|_{\HH_2} }{ \|\delta\|_{\HH_1} } = 0,
\end{equation}
where the limit is with respect to the norm defined on $\HH_1$.

\subsection{Signal Model}

The signal classes we treat are those that are determined by a finite number of parameters per time unit. The $\tau$-local rate of innovation of a signal $x(t)$, denoted $\rho_\tau$, is the minimal number of parameters defining any length-$\tau$ segment of $x(t)$, divided by $\tau$. An FRI signal is one for which $\rho_\tau$ is finite, at least for large enough $\tau$.

Perhaps the simplest class of FRI signals corresponds to functions that can be expressed as
\begin{equation}\label{eq:SI}
x(t) = \sum_{m\in\ZN}a_m g(t-mT)
\end{equation}
with some arbitrary sequence $\{a_m\}\in\ell_2$, where $g(t)$ is a given pulse in $L_2$ and $T>0$ is a given scalar. This set of signals is a linear subspace of $L_2$, which is often termed a \emph{shift-invariant} (SI) space \cite{aldroubi1994}. The subspace of $\pi/T$-bandlimited signals is a special case of \eqref{eq:SI}, with $g(t)=\sinc(t/T)$. Similarly, \eqref{eq:SI} can represent the space of spline functions (by letting $g(t)$ be a B-spline function) and communication signals such as pulse-amplitude modulation (PAM) and quadrature amplitude modulation (QAM). If the support of $g(t)$ is contained in $[t_a,t_b]$, then any interval of the form $[t,t+\tau]$, where $\tau>0$, is affected by no more than $\lceil(t_b-t_a+\tau)/T\rceil$ coefficients from the sequence $\{a_m\}$. This is demonstrated in Fig.~\ref{fig:FRI}\subref{fig:SI_FRI}. Thus, the $\tau$-local rate of innovation of signals of the form \eqref{eq:SI} is
\begin{equation}\label{eq:rhoSI}
\rho_{\tau}=\frac{1}{\tau}\left\lceil\frac{t_b-t_a+\tau}{T}\right\rceil.
\end{equation}
The asymptotic rate of innovation in this case, which can be found by taking $\tau$ to infinity, is $1/T$. We note that, according to our definition, if $g(t)$ is not compactly supported then the rate of innovation is infinite. Thus, for example, bandlimited signals (which correspond to $g(t)=\sinc(t/T)$) are not considered FRI in this paper.

\begin{figure}[t]
\centering
\subfloat[]{\includegraphics[scale=0.85]{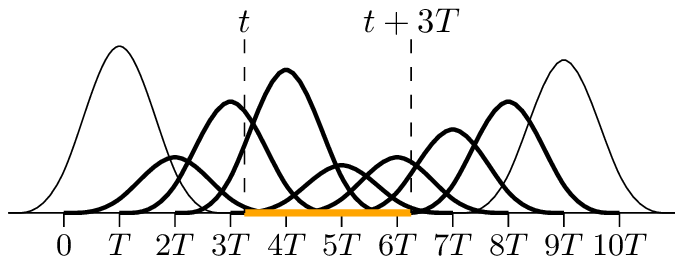}\label{fig:SI_FRI}}\\
\subfloat[]{\includegraphics[scale=0.85]{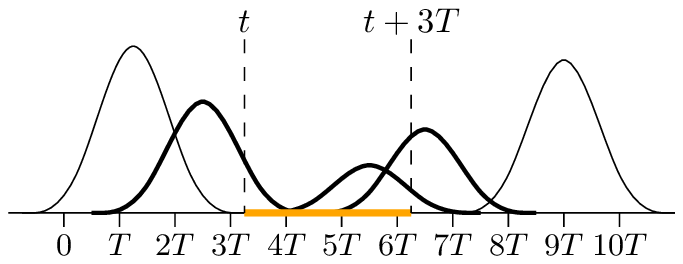}\label{fig:TD_FRI}}
\caption{Streams of shifted versions of a pulse $g(t)$, supported on $[-2T,2T]$. Bold pulses are those that affect the observation segment $[t,t+3T]$. (a)~Fixed pulse positions \eqref{eq:SI}, spaced $T$ seconds apart. Here, the segment $[t,t+3T]$ is affected by $7$ pulses so that $\rho_{3T}=7/(3T)$. (b)~Unknown pulse positions \eqref{eq:ChSounding} with minimal separation $T$. Here, the rate of innovation is $\rho_{3T}=2\times7/(3T)=14/(3T)$. Note that the specific segment $[t,t+3T]$ is affected only by $3$ pulses so that there are $(2\times3)/(3T)=2/T$ parameters per time unit at that location.}
\label{fig:FRI}
\end{figure}

A more complicated model results when the location of the pulses are unknown a-priori, as often happens in channel sounding scenarios. In these cases,
\begin{equation}\label{eq:ChSounding}
x(t) = \sum_{m\in\ZN}a_mg(t-t_m),
\end{equation}
where both $\{a_m\}$ and $\{t_m\}$ are unknown parameters. This class of signals is not a linear subspace, and is therefore much harder to handle. If we fix the time-delays $\{t_m\}$ and vary only the amplitudes $\{a_m\}$ then we get a subspace. But different choices of time-delays result in different subspaces so that overall \eqref{eq:ChSounding} corresponds to a \emph{union of subspaces}. Assuming that the minimal separation between any two time delays is $T$, this model is determined by (at most) twice the number of parameters defining \eqref{eq:SI} per time unit, as demonstrated in Fig.~\ref{fig:FRI}\subref{fig:TD_FRI}. Therefore, the associated $\tau$-local rate of innovation is twice $\rho_{\tau}$ of \eqref{eq:rhoSI} and the asymptotic rate is $2/T$.

The model \eqref{eq:ChSounding} and several of its variants have received the largest amount of attention in the FRI literature\footnote{In fact, the original definition of FRI signals, given in \cite{vetterli2002}, was limited only to functions of the form \eqref{eq:ChSounding}.}. However, other interesting FRI signal classes exist. As an example, suppose that $L$ transmissions of the form \eqref{eq:SI} are modulated, each with a different carrier frequency, to yield
\begin{equation}\label{eq:Multiband}
x(t) = \sum_{\ell=1}^L\sum_{m\in\ZN}a_{\ell,m} g(t-mT)\sin(\omega_\ell t).
\end{equation}
Here, $\{a_{\ell,m}\}_{m\in\ZN}$ is the data transmitted by the $\ell$th user on the carrier frequency $\omega_\ell$. This model generalizes the family of multiband signals treated in \cite{mishali2009,mishali2010}, which corresponds to the case in which $g(t)=\sinc(t/T)$. It is readily seen that if ${\rm supp}\{g\}\in[t_a,t_b]$ then any segment $[t,t+\tau]$ of $x(t)$ is affected by at most $L\lceil(t_b-t_a+\tau)/T\rceil$ of the coefficients $\{a_{\ell,m}\}$. With the addition of the $L$ unknown frequencies, we find that the $\tau$-local rate of innovation of signals of the form \eqref{eq:Multiband} is
\begin{equation}\label{eq:rhoMB}
\rho_{\tau}=\frac{L}{\tau}\left(1+\left\lceil\frac{t_b-t_a+\tau}{T}\right\rceil\right).
\end{equation}
Note that the asymptotic rate of innovation, which is given by $L/T$ in this setting, is not affected by the fact that we do not know the $L$ carrier frequencies. This is because as we increase the observation period, their effect becomes negligible. The set of signals of the form \eqref{eq:rhoMB} is a union of subspaces, where the frequencies $\{\omega_\ell\}$ determine the subspace and the amplitudes $\{a_{\ell,m}\}$ determine the position within the subspace.

To the best of our knowledge, only union-of-subspace settings were treated within the FRI literature. However, FRI signals do not have to conform to the union-of-subspace model. An example is continuous-phase modulation (CPM) transmissions. These include continuous phase frequency shift keying (CPFSK) and minimum shift keying (MSK), tamed frequency modulation (TFM), Gaussian MSK (GMSK) and more. Here, the transmitted signal takes on the form
\begin{equation}\label{eq:CPM}
x(t)=\cos\left(\omega_0 t + 2\pi h\int_{-\infty}^t\sum_{m\in\ZN}a_m g(\tau-mT)d\tau\right),
\end{equation}
where $\omega_0$ is a fixed carrier frequency, $a_m\in\{\pm1,\pm3,\ldots,\pm(Q-1)\}$ are the message symbols, $h$ is the modulation index (usually a rational number), and $g(t)$ is a pulse shape that is supported on $[0,LT]$ for some integer $L>0$ and satisfies $\int_0^{LT} g(t)dt = 0.5$. The rate of innovation of CPM signals can be determined by expressing \eqref{eq:CPM} as
\begin{equation}\label{eq:FSK}
x(t)=\cos\left(\omega_0 t + \sum_{m\in\ZN}\tilde{a}_m \tilde{g}(t-mT)\right),
\end{equation}
where
\begin{equation}
\tilde{a}_m = \sum_{n=-\infty}^m a_n
\end{equation}
and
\begin{equation}
\tilde{g}(t) = 2\pi h \int_{-\infty}^t\left(g(\tau)-g(\tau-T)\right)d\tau.
\end{equation}
Since knowing $\{a_m\}$ is equivalent to knowing $\{\tilde{a}_m\}$ (up to initial boundary condition) and $\tilde{g}(t)$ is supported on $[0,(L+1)T]$, the number of coefficients affecting $x(t)$ on any interval $[t,t+\tau]$ is the same as in \eqref{eq:SI} with $t_a=0$ and $t_b=(L+1)T$. Consequently, the rate of innovation of CPM signals is
\begin{equation}
\rho_{\tau}=\frac{1}{\tau}\left(\left\lceil\frac{\tau}{T}\right\rceil+L+1\right)
\end{equation}
and their asymptotic rate is $1/T$.

Finally, we note that there are union-of-subspace models that do not correspond to FRI signals. As an example, consider the set of signals
\begin{equation}\label{eq:UoSnonFRI}
x(t) = \sum_{m\in\ZN}a_m g_m(t),
\end{equation}
where the only knowledge we have about the pulses $\{g_m(t)\}$ is that they decay exponentially as $t\rightarrow\pm\infty$. Clearly, every possible choice of pulse shapes corresponds to a subspace. Nevertheless, for each $m$, the number of parameters required for describing $g_m(t)$ is infinite.

Any arbitrary length-$\tau$ segment of an FRI signal is determined by at most $K=\lceil\tau\rho_\tau\rceil$ parameters. Therefore, it is reasonable to expect that a properly designed set of $K$ measurements should suffice to identify the parameters of the segment. As discussed in the introduction, this is often the case, implying that many FRI signals can be perfectly recovered from samples taken at their rate of innovation.

Without loss of generality, we focus in this paper on the recovery of an arbitrary segment from an FRI signal. From an abstract viewpoint, any such segment is a vector in some Hilbert space $\HH$, which is known to lie within the set
\begin{equation}\label{eq:setX}
\XX = \left\{x=h(\bth):\bth\in\AAA\right\},
\end{equation}
where $\AAA$ is an open set in $\RN^K$ and $h:\AAA\rightarrow\HH$ is some given function. For example, for any integer $M>0$, the segment $[T+t_b,MT+t_a]$ from \eqref{eq:SI} is affected only by the pulses with indices $m=1,\ldots,M$. Consequently, this segment corresponds to the parameter vector $\bth=\begin{pmatrix}a_1 & \cdots & a_M\end{pmatrix}^T$ and to the function $h:\RN^{M}\rightarrow L_2([T+t_b,MT+t_a])$ given by
\begin{equation}
h : \begin{pmatrix}a_1 & \cdots & a_M\end{pmatrix}^T \mapsto \sum_{m=1}^M a_m g(t-mT).
\end{equation}
Note that, since the signal prior corresponds to a subspace in this case, the function $h$ is linear. In the channel sounding model \eqref{eq:ChSounding}, however, this is no longer true. Specifically, with a minimal separation of $T$ seconds between any two of the time delays $\{t_m\}$, the segment $[T+t_b,MT+t_a]$ from \eqref{eq:ChSounding} is affected by no more than $M$ pulses. Indexing these pulses as $m=1,\ldots,M$, this setting corresponds to the $2M$-dimensional parameter vector $\bth=\begin{pmatrix}t_1 & \cdots & t_M & a_1 & \cdots & a_M\end{pmatrix}^T$ and to the nonlinear function $h:\RN^{2M}\rightarrow L_2([T+t_b,MT+t_a])$ given by
\begin{equation}\label{eq:h_TD}
h : \begin{pmatrix}t_1 & \cdots & t_M & a_1 & \cdots & a_M\end{pmatrix}^T \mapsto \sum_{m=1}^M a_m g(t-t_m).
\end{equation}

We will assume in the sequel that $h$ is Fr\'echet differentiable with respect to the parameter vector $\bth$. This demand is not very restrictive and is satisfied in most practical scenarios. In particular, if the pulse shape $g(t)$ is in $L_2$, then the models \eqref{eq:SI}, \eqref{eq:Multiband} and \eqref{eq:FSK} are all Fr\'echet differentiable with respect to their parameters on any finite-duration interval. If, in addition, $g(t)$ is differentiable and its derivative $g'(t)$ is in $L_2$, then the model \eqref{eq:ChSounding} is also Fr\'echet differentiable. For example, the Fr\'echet derivative of $h$ of \eqref{eq:h_TD} at $\bth_0=\begin{pmatrix}t_1 & \cdots & t_M & a_1 & \cdots & a_M\end{pmatrix}^T$ is the linear operator $(\partial{h}/\partial{\bth})|_{\bth_0} :\RN^{2M}\rightarrow L_2([T+t_b,MT+t_a])$ defined by\footnote{Fr\'echet differentiability is guaranteed in this setting by the fact that the Gateaux (namely directional) derivative of $h$ at $\bth_0$ in the direction $\Delta_{\bth}$ is a bounded linear function of $\Delta_{\bth}$.}
\begin{align}
(\partial{h}/\partial{\bth})|_{\bth_0}\bb &= -a_1g'(t-t_1)b_1 - \cdots -a_Mg'(t-t_M)b_M \nonumber\\
&+ g(t-t_1)b_{M+1} + \cdots + g(t-t_M)b_{2M}.
\end{align}

In addition to the recovery of $x$, it is often of interest to identify the parameters $\bth$ defining it. This goal, of course, cannot be achieved if the parametrization of the set $\XX$ is redundant in the sense that there exist parameters $\bth_1\neq\bth_2$ such that $h(\bth_1)=h(\bth_2)$. To be able to recover $\bth$ in a stable manner, we require the slightly stronger condition that
\begin{equation}\label{eq:h_stability}
\alpha_h\|\bth_1-\bth_2\|_{\RN^K} \leq \|h(\bth_1)-h(\bth_2)\|_{\HH} 
\end{equation}
for some constant $\alpha_h>0$ and for all $\bth_1,\bth_2\in\AAA$. As we discuss in Section~\ref{sec:implicationToUOS} below, some of the aforementioned signal models do not comply with this requirement unless the feasible set $\AAA$ is chosen appropriately.

No further assumptions on the structure of $\XX$, beyond \eqref{eq:h_stability}, are needed for our analysis. Nevertheless, a few remarks are in place regarding the implication of this condition in the widely studied union-of-subspace setting.

\subsection{Implication to Union-of-Subspace Models}
\label{sec:implicationToUOS}

Suppose that $\bth$ can be partitioned as\footnote{The superscripts `N' and `L' stand for nonlinear and linear respectively, intending as a reminder that $h$ is linear in $\bth^{\rm L}$ and nonlinear in $\bth^{\rm N}$.} $\bth=\begin{pmatrix}\bth^{\rm N} & \bth^{\rm L}\end{pmatrix}$,
where the parameters $\bth^{\rm N}$ determine a subspace $\AAA_{\bth^{\rm N}}$ in $\HH$ and the parameters $\bth^{\rm L}$ determine a vector within $\AAA_{\bth^{\rm N}}$. This setting includes as special cases \eqref{eq:ChSounding}, in which $\bth^{\rm N}$ comprises the time shifts $\{t_\ell\}$ and $\bth^{\rm L}$ the amplitudes $\{a_\ell\}$, and \eqref{eq:Multiband}, in which $\bth^{\rm N}$ comprises the frequencies $\{\omega_\ell\}$ and $\bth^{\rm L}$ the sequences $\{a_{\ell,m}\}$.

In this situation, condition \eqref{eq:h_stability} implies that $\bth^{\rm L}$ \emph{must be bounded away from zero} for every signal $x\in\XX$. Indeed, otherwise we could choose $\bth^{\rm L}_1=\bth^{\rm L}_2=\zero$ and $\bth^{\rm N}_1\neq\bth^{\rm N}_2$ so that $h(\bth_1)=h(\bth_2)=0$ despite the fact that $\bth_1\neq\bth_2$.


Condition \eqref{eq:h_stability} also imposes limitations on the parameters $\bth^{\rm N}$. Specifically, assume that the parametrization is such that the subspace $\AAA_{\bth^{\rm N}}$ is not affected by permutation of the elements of $\bth^{\rm N}$. This is the case, for instance, in the channel sounding application \eqref{eq:ChSounding} and in the multiband setting \eqref{eq:Multiband} where $\bth^{\rm N}$ comprises the time delays $\{t_\ell\}$ and frequencies $\{\omega_\ell\}$, respectively. This permutation-invariance implies that if two elements of the vector $\bth^{\rm N}$ are equal, then there exist multiple choices for the parameters $\bth^{\rm L}$ yielding the same signal. Therefore, condition \eqref{eq:h_stability} is clearly violated in this case. We thus conclude that in a permutation-invariant parametrization, the elements of $\bth^{\rm N}$ \emph{must be bounded away from each other}.

Finally, condition \eqref{eq:h_stability} imposes restrictions on the maximal possible distance $\|\bth^{\rm N}_1-\bth^{\rm N}_2\|$ for any two vectors $\bth_1,\bth_2\in\AAA$. More concretely, suppose that the function $h(\bth)$ is such that $\|h(\bth_1)-h(\bth_2)\|$ cannot be made arbitrarily large by varying only the sub-vectors $\bth^{\rm N}_1$ and $\bth^{\rm N}_2$ of $\bth_1$ and $\bth_2$. This always happens, for example, in the channel sounding setting \eqref{eq:ChSounding} with a finitely-supported pulse $g(t)$ because the pulses $g(t-t_1)$ and $g(t-t_2)$ cease to overlap when the distance $|t_2-t_1|$ exceeds the pulse's width. In this setting, condition \eqref{eq:h_stability} cannot be satisfied unless the distance $\|\bth^{\rm N}_1-\bth^{\rm N}_2\|$ is bounded. In other words, $\bth^{\rm N}$ \emph{must be restricted to a bounded set}. Therefore, in model \eqref{eq:ChSounding}, for instance, the time delays $\{t_m\}$ must all lie in some bounded interval. Perhaps a more appealing alternative is to require that $t_1$ lie in some bounded interval and that there exist an upper bound on the separation between any two consecutive time-delays.

To conclude, in the union-of-subspace setting the feasible set $\AAA$ must be such that elements of $\bth^{\rm L}$ are bounded away from zero, the vector $\bth^{\rm N}$ is restricted to a bounded set in $\RN^K$ and its elements are sufficiently separated. This can be achieved in the model \eqref{eq:ChSounding}, for example, by requiring that
\begin{align}\label{eq:ConstrTD}
a_m>a_0, \quad T_{\min} < t_m - t_{m-1} < T_{\max},
\end{align}
for every $m=1\ldots,M$, where $a_0>0$ is a lower-bound on the amplitude, $0<T_{\min}<T_{\max}<\infty$ constitute a lower- and an upper-bound on the separation between consecutive time-delays and $t_0$ is an arbitrary constant.

\subsection{Sampling Method}

Our goal is to recover $x$ by observing $N$ generalized samples $\bc=(c_1,\ldots,c_N)^T$ obtained as
\begin{equation}\label{eq:cSx}
\bc = S(x),
\end{equation}
where $S:\HH\rightarrow \RN^N$ is some (possibly nonlinear) Fr\'echet differentiable operator. This representation is more general than the widely used linear setting, in which
\begin{equation}\label{eq:linSamp}
c_n = \left\langle x, s_n \right\rangle, \quad n=1,\ldots,N,
\end{equation}
for some set of vectors $\{s_n\}_{n=1}^N$ in $\HH$. In particular, \eqref{eq:cSx} may account for nonlinear distortions introduced by the sampling device. For example, $S$ can represent the samples
\begin{equation}\label{eq:NonlinSamp1}
c_n = f(\langle x, s_n \rangle), \quad n=1,\ldots,N,
\end{equation}
where $f(\cdot)$ is a nonlinear sensor response.

We say that a sampling operator $S$ is \emph{stable with respect to} $\XX$ if there exist constants $0<\alpha_s\leq\beta_s<\infty$ such that
\begin{equation}\label{eq:s_stability}
\alpha_s\|x_2-x_1\|_{\HH} \leq \|S(x_1) - S(x_2)\|_{\RN^N} \leq \beta_s\|x_2-x_1\|_{\HH}
\end{equation}
for all $x_1,x_2\in\XX$. This definition is the same as that used in \cite{lu2008} apart from the fact that here the set $\XX$ is not necessarily a union of subspaces and the operator $S$ is not necessarily linear. The left-hand inequality ensures that if two signals $x_1$ and $x_2$ are sufficiently different from one another, then their samples $S(x_1)$ and $S(x_2)$ are different as well. In particular, it implies that two different signals $x_1,x_2\in\XX$ cannot produce the same set of samples, so that there is a unique recovery $x\in\XX$ associated with every valid set of samples $\bc=S(x)\in\RN^N$.

Conditions \eqref{eq:s_stability} and \eqref{eq:h_stability} lie at the heart of any practical sampling theorem, whether implicitly or not. It is out of the scope of this paper to survey the situations in which these conditions are satisfied, as this is rather problem-specific. The interested reader may refer to \cite{dvorkind2008} for an analysis of the SI model \eqref{eq:SI} with nonlinear samples \eqref{eq:NonlinSamp1}, to \cite{blumensath2011} for linear sampling of several union-of-subspace models and for \cite{sun2008} for a general theory for the stability of FRI models. In the sequel we show that these two conditions dictate a minimal sampling rate below which perfect recovery cannot be guaranteed. More interestingly, we will also show that when \eqref{eq:s_stability} and \eqref{eq:h_stability} hold, perfect recovery can be attained at this minimal sampling rate by using a wide family of iterative algorithms.

\section{Minimal Sampling Rate}
\label{sec:SamplingRate}

To be able to devise a general reconstruction strategy for signals in $\XX$ that were sampled by $S$, we first determine the minimal number of samples $N$ required for perfect recovery. Interestingly, conditions \eqref{eq:s_stability} and \eqref{eq:h_stability} implicitly impose a limitation on $N$.
\begin{proposition}\label{lem:NgeqK}
Suppose that the function $h:\AAA\rightarrow \HH$ satisfies \eqref{eq:h_stability} and that the operator $S:\HH\rightarrow\RN^N$ satisfies \eqref{eq:s_stability}. Then
\begin{equation}
N\geq K+\max_{x_1\in\XX}\dim\left(\Nu{\left(\left.\pd{S}{x}\right|_{x_1}\right)^*}\right).
\end{equation}
\end{proposition}

Before providing a proof, we note that Proposition~\ref{lem:NgeqK} shows that the minimal number of samples $N$ required for perfect recovery is the number of parameters $K$ defining $x$. In other words, stable recovery is impossible when sampling below the rate of innovation. While very intuitive and stated in every FRI sampling paper, we believe that this result was not formally proved before for the general signal model and acquisition mechanism discussed in this paper.

Proposition~\ref{lem:NgeqK} further shows that sampling at the rate of innovation is insufficient if the null space of $(\partial{S}/\partial{x})^*$ is nonempty at some $x\in\XX$. When $S$ is a linear operator and $\XX$ is a subspace, spanned by vectors $\{x_k\}_{k=1}^K$, this condition implies that the vectors $\{Sx_k\}_{k=1}^K$ should be linearly independent. In other words, the $N\times K$ matrix whose $(n,k)$ entry is $\langle s_n,x_k\rangle$, should have an empty nullspace, where $\{s_n\}_{n=1}^N$ are the sampling vectors of \eqref{eq:linSamp}. If $S$ is linear but $\XX$ is not contained in any finite-dimensional subspace, then sampling at the rate of innovation necessitates that the sampling vectors $\{s_n\}_{n=1}^N$ be linearly independent. Indeed, if $\{s_n\}_{n=1}^N$ are linearly dependent, then there exists an index $j$ such that $s_j=\sum_{n\neq j}a_ns_n$ for some coefficients $\{a_n\}_{n\neq j}$. Consequently, the sample $c_j$ can be expressed in terms of the other samples as $c_j=\langle x,s_j\rangle = \sum_{n\neq j}\bar{a}_n\langle x,s_n\rangle = \sum_{n\neq j}\bar{a}_nc_n$ and thus can be disregarded.

As another example, suppose that one of the measurements produced by the sensing device, say $c_1$, is the energy $0.5\|x\|^2$ of $x$. In this case $(\partial{c_1}/\partial{x})|_{x_1}=x_1$. Consequently, from Proposition~\ref{lem:NgeqK}, sampling at the minimal rate is impossible if the set of signals $\XX$ contains the signal $x_1=0$. The intuition here follows from the observation that small perturbations in $x$ around the signal $x_1=0$ do not show in $c_1$. Therefore, if the input to our sampling device happens to be $x=0$ in this setting, then sampling is unavoidably unstable, as the left-hand side of condition \eqref{eq:s_stability} is violated.

\begin{proof}
Since $h(\bth)$ and $S(x)$ are Fr\'echet differentiable, it follows that the function $\hat{\bc}(\bth)=S(h(\bth))$ is Fr\'echet differentiable as well. We will start by showing that its derivative $\partial{\hat{\bc}}/\partial{\bth}$, which is an $N\times K$ matrix, has an empty null space.

By definition, the Fr\'echet derivative $\partial{\hat{\bc}}/\partial{\bth}$ at $\bth_1$ satisfies
\begin{equation}
\lim_{\bdelta\rightarrow 0} \frac{\left\|\hat{\bc}(\bth_1+\bdelta)-\hat{\bc}(\bth_1)-\left.\pd{\hat{\bc}}{\bth}\right|_{\bth_1}\bdelta\right\|_{\RN^N}} {\|\bdelta\|_{\RN^K}} = 0.
\end{equation}
In particular, for any nonzero $\ba\in\RN^K$,
\begin{equation}\label{eq:NgeqK_proof1}
\lim_{t\rightarrow 0} \frac{\left\|\hat{\bc}(\bth_1+t\ba)-\hat{\bc}(\bth_1)-t\left.\pd{\hat{\bc}}{\bth}\right|_{\bth_1}\ba\right\|_{\RN^N}} {\|t\ba\|_{\RN^K}} = 0,
\end{equation}
where $t$ is a scalar variable. Now, assume that $\ba\in\Nu{\partial{\hat{\bc}}/\partial{\bth}|_{\bth_1}}$. Then \eqref{eq:NgeqK_proof1} implies that
\begin{equation}\label{eq:NgeqK_proof2}
\lim_{t\rightarrow 0} \frac{\left\|\hat{\bc}(\bth_1+t\ba)-\hat{\bc}(\bth_1)\right\|_{\RN^N}} {\|t\ba\|_{\RN^K}} = 0.
\end{equation}
However, \eqref{eq:h_stability} and \eqref{eq:s_stability} imply that
\begin{align}
&\frac{\left\|\hat{\bc}(\bth_1+t\ba)-\hat{\bc}(\bth_1)\right\|_{\RN^N}} {\|t\ba\|_{\RN^K}}
= \frac{\left\|S(h(\bth_1+t\ba))-S(h(\bth_1))\right\|_{\RN^N}} {\|t\ba\|_{\RN^K}} \nonumber\\
&\hspace{3cm}\geq \alpha_s\frac{\left\|h(\bth_1+t\ba)-h(\bth_1)\right\|_{\RN^N}} {\|t\ba\|_{\RN^K}} \nonumber\\
&\hspace{3cm}\geq \alpha_s\alpha_h > 0
\end{align}
for every $t\neq 0$. This contradicts \eqref{eq:NgeqK_proof2} and therefore demonstrates that $\mathcal{N}(\partial{\hat{\bc}}/\partial{\bth}|_{\bth_1})=\{\zero\}$, which implies that $\dim(\mathcal{R}(\partial{\hat{\bc}}/\partial{\bth}|_{\bth_1}))= K$.

Next, note that $\partial{\hat{\bc}}/\partial{\bth}|_{\bth_1}=(\partial{S}/\partial{x}|_{h(\bth_1)})(\partial{h}/\partial{\bth}|_{\bth_1})$ so that $\mathcal{R}(\partial{\hat{\bc}}/\partial{\bth}|_{\bth_1})\subseteq \mathcal{R}(\partial{S}/\partial{x}|_{h(\bth_1)})= \mathcal{N}((\partial{S}/\partial{x}|_{h(\bth_1)})^*)^\perp$. Therefore,
\begin{align}\label{eq:K}
K &= \dim\left(\Ra{\left.\pd{\hat{\bc}}{\bth}\right|_{\bth_1}}\right) \nonumber\\
&\leq \dim\left(\Nu{\left(\left.\pd{S}{x}\right|_{h(\bth_1)}\right)^*}^\perp\right) \nonumber\\
&= N - \dim\left(\Nu{\left(\left.\pd{S}{x}\right|_{h(\bth_1)}\right)^*}\right).
\end{align}
Since \eqref{eq:K} holds for every $\bth_1\in\AAA$, it holds for the $\bth_1$ minimizing the right-hand side, completing the proof.
\end{proof}

Throughout the rest of the paper we focus on the case in which $N=K$ samples of $x(t)$ are obtained with an operator $S$ satisfying
\begin{equation}\label{eq:empthyNulspace}
\Nu{\left(\left.\pd{S}{x}\right|_{x_1}\right)^*} = \{\zero\},\quad \forall x_1\in\XX.
\end{equation}
This corresponds to sampling at the rate of innovation.

\section{Least Squares Recovery}
\label{sec:LSrecovery}

Suppose we want to recover a signal $x=h(\bth_0)\in\HH$ from its samples $\bc=S(x)$, where $\bth_0\in\RN^K$ is an unknown parameter vector and $S:\HH\rightarrow\RN^K$ is a given sampling operator. To address this problem, it is natural to seek the minimizer of the function
\begin{equation}\label{eq:ls}
\varepsilon(\bth) = \frac 1 2 \|S(h(\bth))-\bc\|_{\RN^K}^2 = \frac 1 2 \|\hat{\bc}(\bth)-\bc\|_{\RN^K}^2,
\end{equation}
where we defined $\hat{\bc}(\bth)=S(h(\bth))$. The reasoning behind this choice follows from the following observation
\begin{proposition}\label{lem:UniqueMinimizer}
Suppose that the function $h:\RN^K\rightarrow \HH$ satisfies \eqref{eq:h_stability} and that the operator $S:\HH\rightarrow\RN^K$ satisfies \eqref{eq:s_stability}. Then $\bth_0$ is the unique global minimizer of $\varepsilon(\bth)$.
\end{proposition}
\begin{proof}
Clearly, $\varepsilon(\bth)\geq0$ for every $\bth\in\RN^K$ and $\varepsilon(\bth_0)=0$, so that $\bth_0$ is a global minimizer of $\varepsilon(\bth)$. This minimizer is unique since, due to \eqref{eq:h_stability} and \eqref{eq:s_stability}, $\varepsilon(\bth)\geq\alpha_s\alpha_h\|\bth-\bth_0\|_{\RN^K}$ so that $\varepsilon(\bth)>0$ for every $\bth\neq\bth_0$.
\end{proof}

The LS criterion \eqref{eq:ls} is also plausible when the samples $\bc$ correspond to a perturbation of the true sample vector by white Gaussian noise. In this case, the minimizer of \eqref{eq:ls} is a maximum-likelihood estimate of $\bth$ from $\bc$.


Unfortunately, the function $\varepsilon(\bth)$ is generally non-convex and might possess many local minima. It therefore seems that standard optimization techniques may fail in finding its global minimizer $\bth_0$. However, as we show next, when sampling at the rate of innovation, assumptions \eqref{eq:h_stability} and \eqref{eq:s_stability} guarantee that $\bth_0$ is the unique stationary point of $\varepsilon(\bth)$. Thus, any algorithm designed to trap a stationary point, necessarily converges to the true parameter vector $\bth_0$. The proof of this result follows that of \cite[Theorem~6]{dvorkind2008}, which treats the special case of SI signals and memoryless nonlinear samples.

\begin{theorem}\label{thm:UniqeStationaryPoint}
Suppose that the function $h:\RN^K\rightarrow \HH$ satisfies \eqref{eq:h_stability}, the operator $S:\HH\rightarrow\RN^K$ satisfies \eqref{eq:s_stability} and its  Fr\'{e}chet derivative $\partial{S}/\partial{x}$ satisfies \eqref{eq:empthyNulspace}. Then $\nabla\varepsilon(\bth_1)=\zero$ only if $\bth_1=\bth_0$.
\end{theorem}
\begin{proof}
The gradient $\nabla\varepsilon(\bth_1)$ is given by
\begin{equation}\label{eq:gradC}
\nabla\varepsilon(\bth_1) = \left(\left.\pd{\hat{\bc}}{\bth}\right|_{\bth_1}\right)^* \left(\hat{\bc}(\bth_1)-\bc\right).
\end{equation}
We showed in the proof of Proposition~\ref{lem:NgeqK} that $\Ra{\partial{\hat{\bc}}/\partial{\bth}|_{\bth_1}} = \RN^K$. Since here $\partial{\hat{\bc}}/\partial{\bth}|_{\bth_1}$ is a $K\times K$ matrix, it follows that
\begin{equation}
\Nu{\left(\left.\pd{\hat{\bc}}{\bth}\right|_{\bth_1}\right)^*} = \Ra{\left.\pd{\hat{\bc}}{\bth}\right|_{\bth_1}}^\perp = \{\zero\},
\end{equation}
so that $\nabla\varepsilon(\bth_1)=0$ only if $\hat{\bc}(\bth_1)-\bc=\zero$. This, by Proposition~\ref{lem:UniqueMinimizer}, happens only if $\bth_1=\bth_0$, completing the proof.
\end{proof}

The importance of Theorem~\ref{thm:UniqeStationaryPoint} lies in the fact that it provides a unified mechanism for recovering FRI signals from samples taken at the rate of innovation. Namely, rather than developing a different algorithm for every choice of signal family and sampling method, we can employ the same general-purpose optimization technique to find the stationary point of \eqref{eq:ls}. Furthermore, this strategy is also advantageous over the iterative approach of \cite{blumensath2011}, as it avoids the need for projecting the signal estimate onto $\XX$ in each iteration, an operation that possesses no closed form solution for most FRI signal classes.

\section{Iterative Recovery}
\label{sec:IterativeRecovery}

There are numerous optimization algorithms that can be used to find the stationary point of the objective function $\varepsilon(\bth)$ over $\AAA$. For simplicity, we focus here on unconstrained optimization methods, namely those that can be applied when $\AAA=\RN^K$. This does not limit the generality of the discussion since if $\AAA\neq\RN^K$, then the constrained problem $\min_{\bth\in\AAA}\varepsilon(\bth)$ can be transformed into the unconstrained problem $\min_{\tilde{\bth}\in\RN^K}\varepsilon(p(\tilde{\bth}))$, where $p:\RN^K\rightarrow\AAA$ is one-to-one and onto. The latter problem possesses a unique stationary point $\tilde{\bth}_0=p^{-1}(\bth_0)$. Therefore, once $\tilde{\bth}_0$ is determined, the desired solution can be computed as $\bth_0=p(\tilde{\bth}_0)$. For example, the model \eqref{eq:ChSounding} with the set $\AAA$ of constraints defined by \eqref{eq:ConstrTD} can be handled by defining
\begin{equation}\label{eq:ParamThetaNew}
\tilde{\theta}^{\rm L}_m = \ln(a_m-a_0),\quad
\tilde{\theta}^{\rm N}_m = \tan\left(\pi\frac{t_m-t_{m-1}-\bar{T}}{\Delta}\right),
\end{equation}
where $\bar{T}=(T_{\max}+T_{\min})/2$ and $\Delta=T_{\max}-T_{\min}$, so that
\begin{equation}\label{eq:ParamThetaNewInverse}
a_m=e^{\tilde{\theta}^{\rm L}_m}+a_0, \quad t_m = t_0+m\bar{T}+\frac{\Delta}{\pi}\sum_{i=1}^m \arctan\left(\tilde{\theta}^{\rm N}_i\right).
\end{equation}
With this choice, the set $\XX$ of all feasible signals is obtained by varying $\tilde{\bth}^{\rm L}$ and $\tilde{\bth}^{\rm N}$ over the entire space $\RN^M$ and not over some subset of $\RN^M$.

Most unconstrained optimization methods start with an initial guess $\bth^0$ and perform iterations of the form
\begin{equation}\label{eq:iterGeneral}
\bth^{\ell+1}=\bth^{\ell} - \gamma^\ell \bB^\ell \nabla\varepsilon(\bth^\ell),
\end{equation}
where $\gamma^\ell$ is a scalar step size obtained by means of a one dimensional search and $\bB^\ell$ is a positive definite matrix. Due to the structure of $\nabla\varepsilon(\bth^\ell)$ in our case (see \eqref{eq:gradC}), the iterations \eqref{eq:iterGeneral} can be given a simple interpretation, as shown in Fig.~\ref{fig:Alg}. Specifically, at the $\ell$th iteration, the current estimate $\bth^\ell$ of the parameters $\bth$ is used to construct our estimate $\hat{x}^\ell$ of the signal $x$ by applying the function $h$. This estimate is then sampled using the operator $S$ to obtain an estimated sample vector $\hat{\bc}^\ell$. Finally, the difference between $\hat{\bc}^\ell$ and the true set of samples $\bc$ is multiplied by a correction matrix and added to $\bth^\ell$ to yield the updated estimate $\bth^{\ell+1}$ of the parameter vector $\bth$.
\begin{figure}[t]
\centering
\includegraphics[scale=0.85]{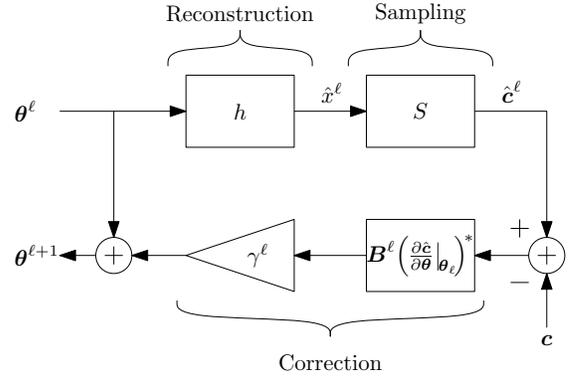}
\caption{Schematic interpretation of one iteration of \eqref{eq:iterGeneral}.}
\label{fig:Alg}
\end{figure}

In our setting, the objective function $\varepsilon(\bth)$ is bounded from below. The iterations $\eqref{eq:iterGeneral}$ are therefore guaranteed to converge to a stationary point of $\varepsilon(\bth)$ if $\gamma^\ell$ is chosen to satisfy the Wolfe conditions \cite{nocedal1999}, $\bB^\ell$ is chosen such that
\begin{equation}\label{eq:cos}
\frac{\left\langle\bB^\ell \nabla\varepsilon(\bth^\ell),\nabla\varepsilon(\bth^\ell)\right\rangle_{\RN^K}}{ \left\|\bB^\ell \nabla\varepsilon(\bth^\ell)\right\|_{\RN^K}\left\|\nabla\varepsilon(\bth^\ell)\right\|_{\RN^K}} > \delta
\end{equation}
for some constant $\delta>0$ independent of $\ell$, and the gradient $\nabla\varepsilon(\bth)$ is Lipschitz continuous in an environment of the level-set $\mathcal{N}=\{\bth:\varepsilon(\bth)\leq\varepsilon(\bth^0)\}$ \cite{nocedal1999}.

A step size satisfying the Wolfe conditions can be found by using the backtracking method \cite{nocedal1999}, as presented in Algorithm~\ref{alg:Backtracking}. Condition \eqref{eq:cos} is trivially satisfied with $\bB^\ell=\bI$, which corresponds to the steepest descent method. As we show in Appendix~\ref{sec:ConvNewton}, this condition is also satisfied with $\bB^\ell=((\partial{\hat{\bc}}/\partial{\bth}|_{\bth^\ell})^*(\partial{\hat{\bc}}/\partial{\bth}|_{\bth^\ell}))^{-1}$ if
\begin{equation}\label{eq:beta_h}
\|h(\bth_1)-h(\bth_2)\|_{\HH} \leq \beta_h\|\bth_1-\bth_2\|_{\RN^K}
\end{equation}
for some $\beta_h<\infty$ and for all $\bth_1,\bth_2\in\mathcal{N}$. This choice belongs to the class of quasi-Newton methods, which typically converge much faster than steepest descent. Finally, we show in Appendix~\ref{sec:GradLip} that a sufficient condition for $\nabla\varepsilon(\bth)$ to be Lipschitz continuous over $\mathcal{N}$ is that the derivative of $h$ be Lipschitz continuous there, namely that
\begin{equation}\label{eq:hprime_stability}
\left\|\left.\pd{h}{\bth}\right|_{\bth_1}-\left.\pd{h}{\bth}\right|_{\bth_2}\right\| \leq \beta_{h'}\|\bth_1-\bth_2\|_{\RN^K}
\end{equation}
for some $\beta_{h'}<\infty$ and for all $\bth_1,\bth_2\in\mathcal{N}$. The analyses in Appendices~\ref{sec:ConvNewton} and \ref{sec:GradLip} follow closely those in the proof of \cite[Theorem~7]{dvorkind2008}. To summarize, we have the following result.
\begin{theorem}
Suppose that the function $h:\RN^K\rightarrow \HH$ satisfies \eqref{eq:h_stability}, its Fr\'{e}chet derivative $\partial{h}/\partial{\bth}$ satisfies \eqref{eq:hprime_stability}, the operator $S:\HH\rightarrow \RN^K$ satisfies \eqref{eq:s_stability} and its Fr\'{e}chet derivative $\partial{S}/\partial{x}$ satisfies \eqref{eq:empthyNulspace}. Consider the iterations \eqref{eq:iterGeneral}, where the step size $\gamma^\ell$ is obtained via Algorithm~\ref{alg:Backtracking}. Then each of the following options guarantees that $\bth^\ell\rightarrow\bth_0$:
\begin{enumerate}
\item $\bB^\ell=\bI$.
\item $\bB^\ell=((\partial{\hat{\bc}}/\partial{\bth}|_{\bth^\ell})^*(\partial{\hat{\bc}}/\partial{\bth}|_{\bth^\ell}))^{-1}$ and condition \eqref{eq:beta_h} holds.
\end{enumerate}
\end{theorem}

\begin{algorithm}[t]\caption{Backtracking line search.}
\begin{algorithmic}
\STATE{set $\bg^\ell=\nabla\varepsilon(\bth^\ell)$, $\bd^\ell=-\bB^\ell\bg^\ell$, $\delta = 1$ and $\rho,\eta\in(0,1)$}
\WHILE{$\varepsilon(\bth^\ell+\delta\bd^\ell) > \varepsilon(\bth^\ell)+\eta\delta\langle \bd^\ell, \bg^\ell \rangle_{\RN^K}$}
\STATE{$\delta \leftarrow \rho\delta$}
\ENDWHILE
\RETURN{$\gamma^\ell = \delta$}
\end{algorithmic}
\label{alg:Backtracking}
\end{algorithm}

\section{Application to Channel Sounding}
\label{sec:ApplicationChannelSounding}

We now demonstrate our approach in the channel sounding setting \eqref{eq:ChSounding}. Specifically, suppose that
\begin{equation}\label{eq:ChannelSoundingM}
x(t) = \sum_{m=1}^M a_m g(t-t_m),\quad t\in[0,\tau],
\end{equation}
where $g(t)$ is a known pulse shape, $\{a_m\}_{m=1}^M$ are unknown amplitudes, and $\{t_m\}_{m=1}^M$ are unknown time-delays. As explained in Section~\ref{sec:implicationToUOS}, the parameter vector
\begin{equation}
\bth=
\begin{pmatrix}t_1 & \cdots & t_M & a_1 & \cdots & a_M\end{pmatrix}^T,
\end{equation}
cannot be stably recovered unless the amplitudes all surpass a certain threshold and the pulses are well separated yet confined to a bounded interval. We therefore adopt the assumptions \eqref{eq:ConstrTD} and transform the optimization problem into an unconstrained one by using the parameter vector $\tilde{\bth}=p^{-1}(\bth)$ described in \eqref{eq:ParamThetaNew}, with the transformation $\bth=p(\tilde{\bth})$ of \eqref{eq:ParamThetaNewInverse}. 
Our goal is to recover the signal parameters from the samples \eqref{eq:NonlinSamp1}, where $\{s_n(t)\}_{n=1}^N$ are sampling kernels in $L_2([0,\tau])$ and $f(\cdot)$ is a nonlinear response function.

As discussed in the introduction, when $f(\cdot)$ is the identity operator, there are several combinations of pulse shapes $g(t)$ and sampling kernels $\{s_n(t)\}$ that can be treated via existing algorithms in a stable manner, such as \cite{tur2011,gedalyahu2011}. However, none of the existing techniques is applicable when $f(\cdot)$ is nonlinear. Furthermore, as we demonstrate in Section~\ref{sec:GaussiansNonlin} below, our approach allows recovery from SI samples with a kernel that is not supported by \cite{tur2011}. Moreover, in Section~\ref{sec:KfirRonen} we apply our technique in a multichannel setting for which the algorithm of \cite{gedalyahu2011} is applicable, and show the advantage of our approach in the presence of noise.

To apply the quasi-Newton or steepest decent methods, we note that, with the transformation $\bth=p(\tilde{\bth})$ of \eqref{eq:ParamThetaNewInverse},
\begin{equation}\label{eq:chainRule}
\pd{\hat{\bc}}{\tilde{\bth}} = \pd{\hat{\bc}}{\bth}\pd{p}{\tilde{\bth}}.
\end{equation}
Explicit computation shows that
\begin{align}
\pd{\hat{\bc}}{\bth} = \bC \begin{pmatrix}\bA & \bB \end{pmatrix}
\end{align}
with
\begin{align}\label{eq:matA}
\bA =
\begin{pmatrix}
-a_1\langle g'(\cdot-t_1), s_1\rangle & \cdots & -a_M\langle g'(\cdot-t_M), s_1\rangle \\
\vdots & & \vdots\\
-a_1\langle g'(\cdot-t_1), s_N\rangle & \cdots & -a_M\langle g'(\cdot-t_M), s_N\rangle
\end{pmatrix},
\end{align}
\begin{align}\label{eq:matB}
\bB =
\begin{pmatrix}
\langle g(\cdot-t_1), s_1\rangle & \cdots & \langle g(\cdot-t_M), s_1\rangle \\
\vdots & & \vdots \\
\langle g(\cdot-t_1), s_N\rangle & \cdots & \langle g(\cdot-t_M), s_N\rangle
\end{pmatrix}
\end{align}
and
\begin{align}
\bC = {\rm diag}
\begin{pmatrix}
f'(\langle x, s_1\rangle) & \cdots & f'(\langle x, s_N\rangle).
\end{pmatrix}
\end{align}
Furthermore,
\begin{align}
\pd{p}{\tilde{\bth}} = \begin{pmatrix} \bD & \bzero \\ \bzero & \bE \end{pmatrix}
\end{align}
with
\begin{align}
\bD = {\rm diag}
\begin{pmatrix}
e^{\tilde{\theta}_1} & \cdots & e^{\tilde{\theta}_M}
\end{pmatrix}
\end{align}
and
\begin{align}
\bE = \frac{\Delta}{\pi}
\begin{pmatrix}
\frac{1}{1+\tilde{\theta}_{M+1}^2} & 0                                  &  \cdots & 0 \\
\frac{1}{1+\tilde{\theta}_{M+1}^2} & \frac{1}{1+\tilde{\theta}_{M+2}^2} &  \cdots & 0 \\
\vdots                             & \vdots                             &  \ddots & \vdots\\
\frac{1}{1+\tilde{\theta}_{M+1}^2} & \frac{1}{1+\tilde{\theta}_{M+2}^2} &  \cdots & \frac{1}{1+\tilde{\theta}_{2M}^2}
\end{pmatrix}.
\end{align}

We now demonstrate our method in several specific settings.

\subsection{Gaussian Pulses and Gaussian Kernels with Nonlinear Distortion}
\label{sec:GaussiansNonlin}

Consider the sampling system of Fig.~\ref{fig:NonlinearSampling}, in which $x(t)$ is sampled after passing through an amplitude limiter $f(\cdot)$ and being convolved with a filter $s(-t)$. The resulting samples can be described by \eqref{eq:NonlinSamp1}, with $s_n(t)=s(t-T_0-nT_{\rm s})$. Since the model \eqref{eq:ChannelSoundingM} is clearly determined by $K=2M$ parameters, we would like to recover any such $x(t)$ from $N=2M$ samples. We choose the sampling period $T_{\rm s}$ to equal $\tau/N$ and the offset $T_0$ to be $T_{\rm s}/2$, so that the sampling functions span the entire observation segment $[0,\tau]$.

\begin{figure}[t]\centering
\includegraphics[scale=0.85]{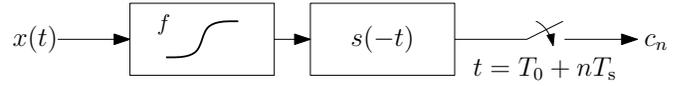}
\caption{Nonlinear and nonideal sampling.}
\label{fig:NonlinearSampling}
\end{figure}

Figure~\ref{fig:GaussiansNonlinear} demonstrates the convergence of the Newton iterations for recovering $M=2$ pulses over the period $[0,1]$ from $N=4$ samples. Here, the pulse shape and the sampling filter were taken to be Gaussian functions with variances $\sigma_g=0.05$ and $\sigma_s=0.1$, respectively.
Note that, with this choice, all inner products in \eqref{eq:matA} and \eqref{eq:matB} can be computed analytically at every iteration. The nonlinear response curve was set to be $f(c)=100 \arctan(0.01 c)$. The constraints \eqref{eq:ConstrTD} we assumed on the parameters corresponded to $a_0=0.1$, $T_{\min}=0.3$, $T_{\max}=0.7$ and $t_0=-0.3$.

The true parameters in this experiment were $t_1 =0.2$, $t_2=0.8$, $a_1=1$ and $a_2=5$. As shown in Fig.~\ref{fig:GaussiansNonlinear}\subref{fig:iteration0}, the iterations were initialized at $t_1=1/3$, $t_2=2/3$ and $a_1=a_2=3$. The estimated samples at this point, shown in `x'-marks, deviate substantially from the true samples, marked with circles. As can be seen, though, this gap decreases quickly in the first $15$ iterations (see Fig~\ref{fig:GaussiansNonlinear}\subref{fig:iteration15}) and almost completely vanishes after $30$ iterations (Fig~\ref{fig:GaussiansNonlinear}\subref{fig:iteration30}). Figure~\ref{fig:GaussiansNonlinear}\subref{fig:ObjVsIt} shows the rapid decrease in the LS objective \eqref{eq:ls} as a function of the iterations.

\begin{figure}[t]\centering
\subfloat[]{\includegraphics[scale=0.4]{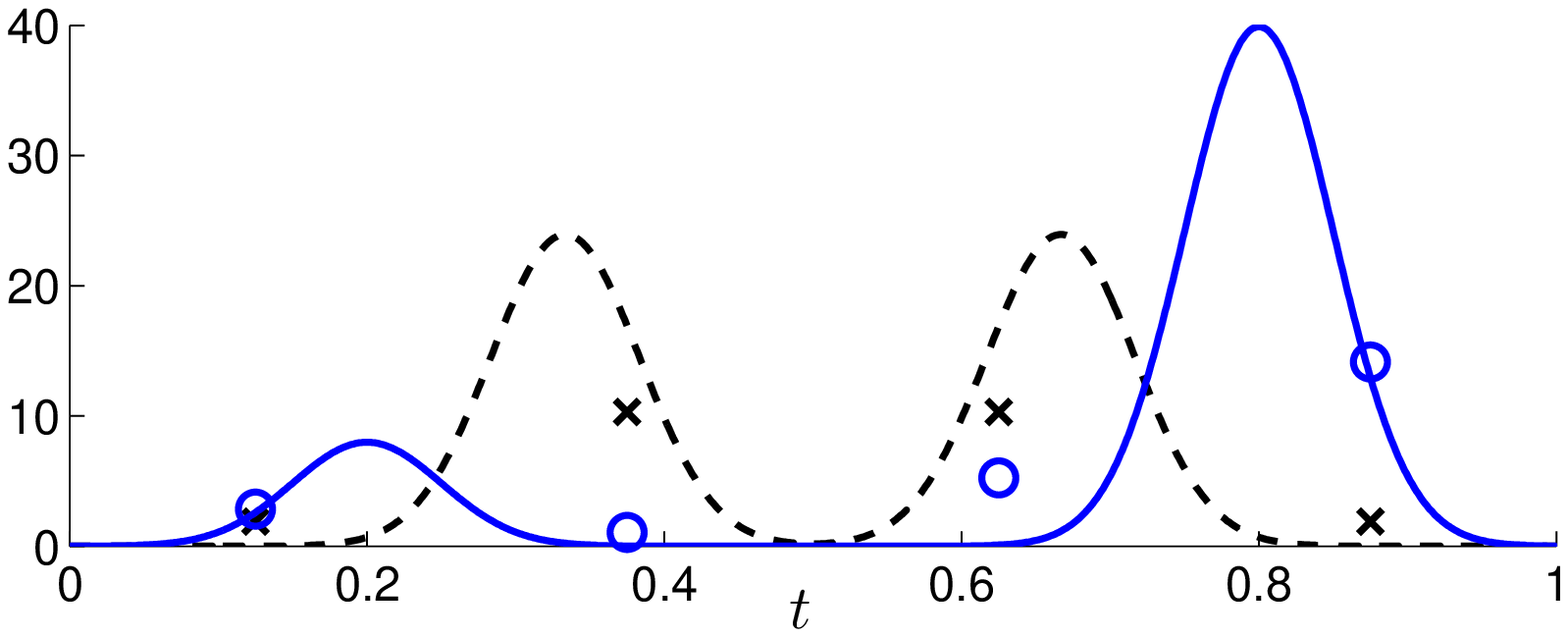}\label{fig:iteration0}}\\
\subfloat[]{\includegraphics[scale=0.4]{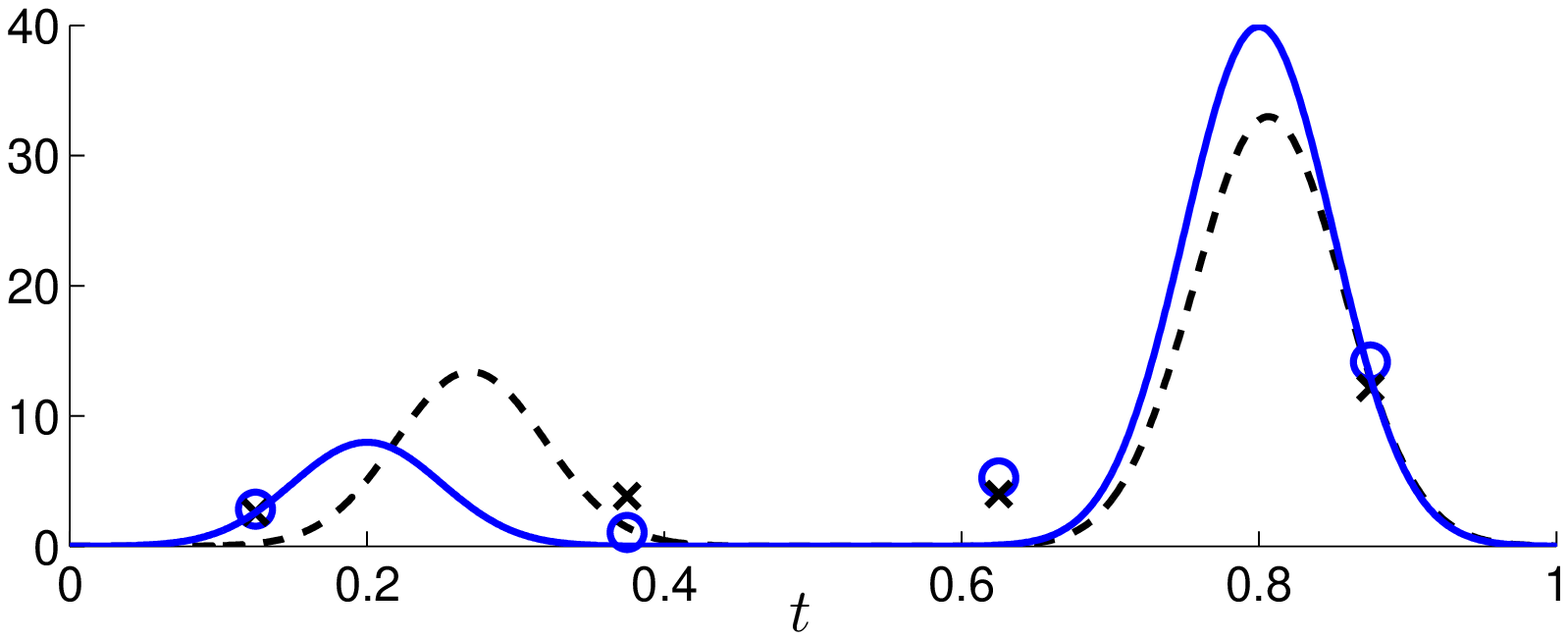}\label{fig:iteration15}}\\
\subfloat[]{\includegraphics[scale=0.4]{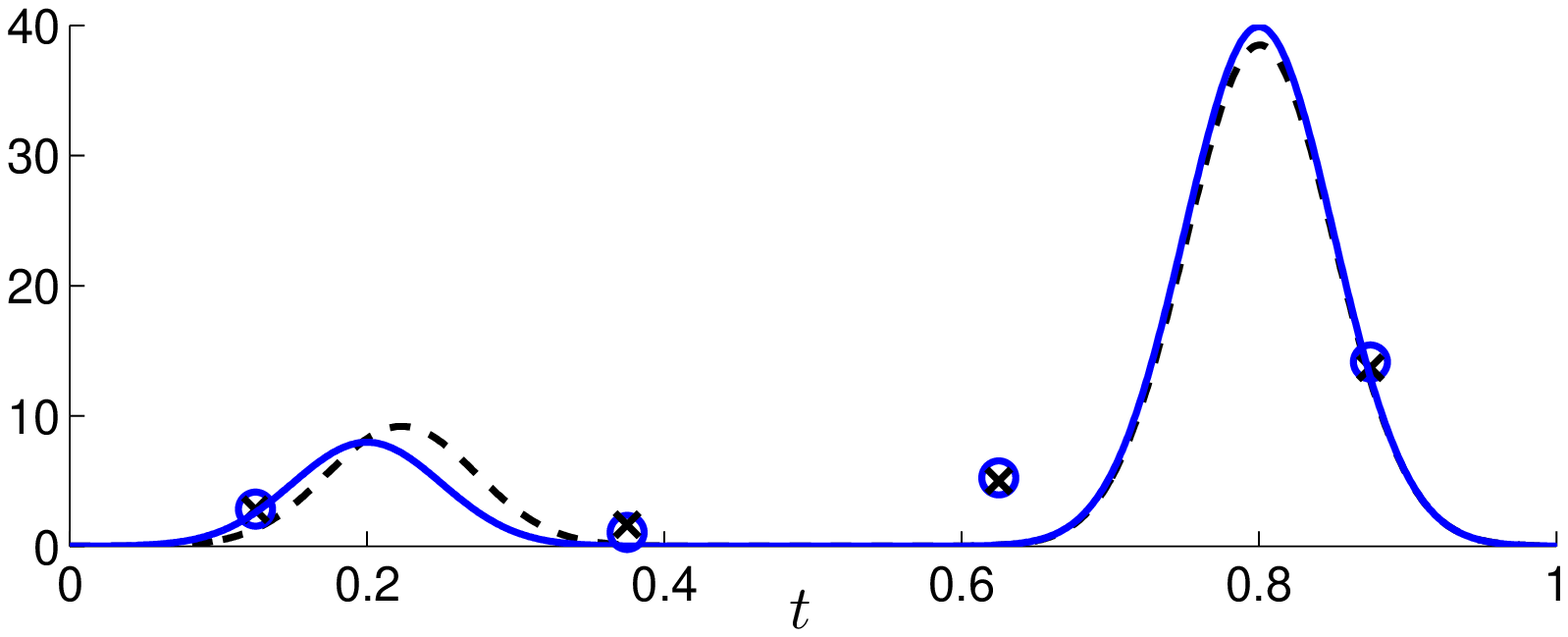}\label{fig:iteration30}}\\
\subfloat[]{\includegraphics[scale=0.4]{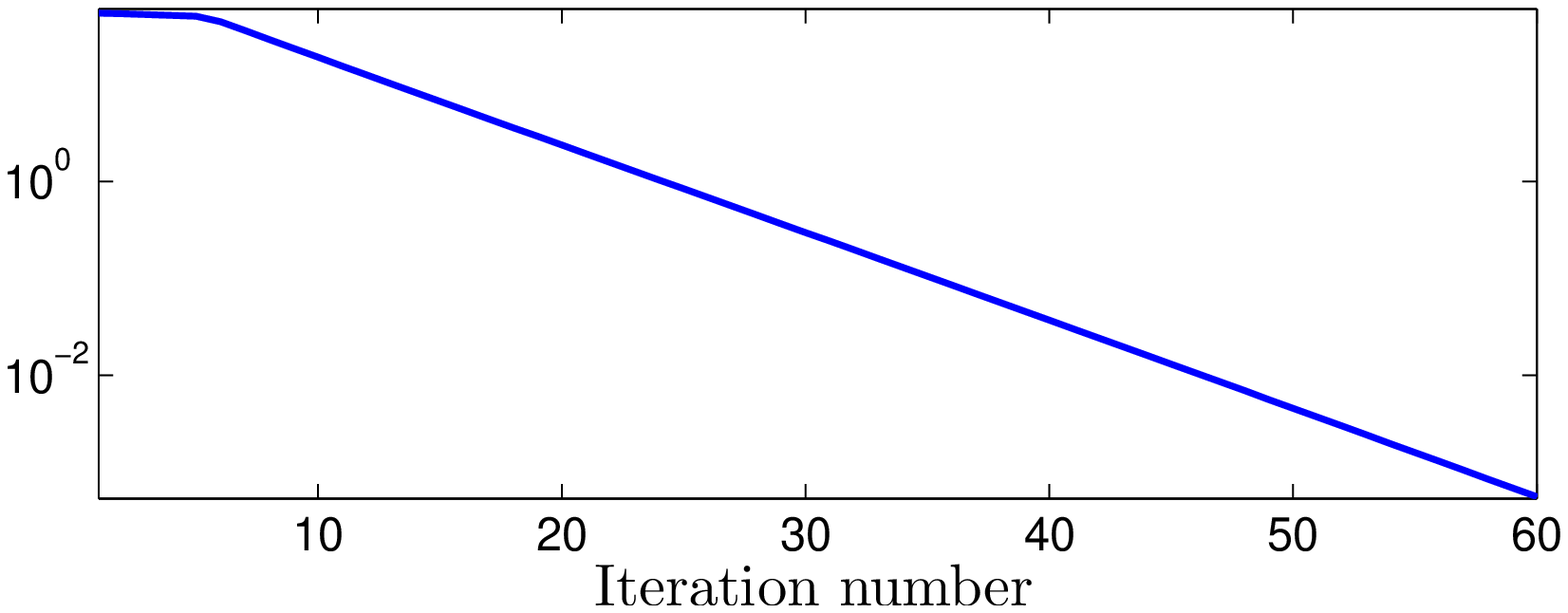}\label{fig:ObjVsIt}}
\caption{Convergence of Newton iterations for pulse stream recovery. (a)~Initialization. (b)~$15$th iteration. (c)~$30$th iteration. (d)~LS objective value as a function of the iterations.}
\label{fig:GaussiansNonlinear}
\end{figure}

Figure~\ref{fig:MSEvsSNRGaussiansNonlinear} demonstrates the behavior of the algorithm in the presence of noise. The setting here is the same as that of Fig.~\ref{fig:GaussiansNonlinear} with the distinction that white Gaussian noise is added to the samples prior to recovery. This figure depicts the mean squared error (MSE) in $x(t)$, defined as
\begin{equation}
{\rm MSE}=\EE\!\left[\int_0^\tau \left|x(t)-\hat{x}(t)\right|^2 dt\right],
\end{equation}
as a function of the signal-to-noise (SNR) ratio. The solid line corresponds to the Cram\'{e}r-Rao bound (CRB), developed in \cite{ben-haim2011}, which is a lower bound on the MSE attainable by any unbiased estimation technique. As can be seen, the MSE of our method coincides with the CRB in high SNR scenarios and outperforms at in low SNR levels. This is a result of the fact that our technique is biased.

\begin{figure}[t]\centering
\includegraphics[scale=0.4, trim=0cm 0.1cm 0cm 0.1cm, clip]{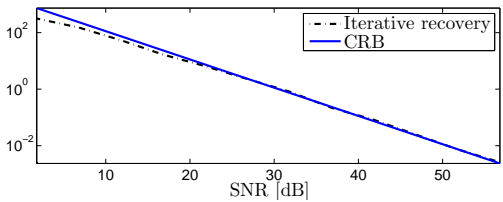}
\caption{MSE as a function of SNR for pulse stream recovery in the setting of Fig.~\ref{fig:GaussiansNonlinear}.}
\label{fig:MSEvsSNRGaussiansNonlinear}
\end{figure}

\subsection{Periodic Pulses and Sinusoidal Kernels}
\label{sec:KfirRonen}

Next, we turn to demonstrate our approach in a periodic pulse-stream scenario with the multichannel sampling system of \cite{gedalyahu2011}. Specifically, suppose that $g(t)$ in \eqref{eq:ChannelSoundingM} is $\tau$-periodic with Fourier coefficients $\tilde{g}_k = (1/\tau)\langle g,\phi_k\rangle$, where $\phi_k(t)=e^{2\pi j k t/\tau}$. In \cite{gedalyahu2011}, it was shown that the pulse parameters can be identified in this setting by using the multichannel sampling system depicted in Fig.~\ref{fig:MultichannelSampling}, where the sampling kernels $s_n(t)$ correspond to combinations of the complex exponentials $\{\phi_k(t)\}_{k\in\mathcal{K}}$ with $\mathcal{K}$ being a set of consecutive indices. The algorithm of \cite{gedalyahu2011} was developed for linear sampling, so that $f(\cdot)$ of \eqref{eq:NonlinSamp1} is set to be the identity. This algorithm is based on applying techniques for identifying the frequencies of complex exponentials, such as the matrix pencil \cite{hua90} or annihilating filter \cite{vetterli2002} methods.

\begin{figure}[t]\centering
\includegraphics[scale=0.85]{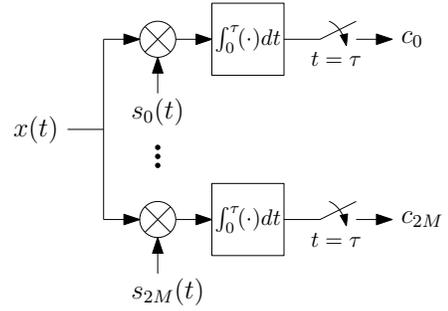}
\caption{Linear multichannel sampling.}
\label{fig:MultichannelSampling}
\end{figure}

If we restrict attention to real sampling functions, then the minimal number $N$ of samples supported by the method of \cite{gedalyahu2011} is $2M+1$. This is achieved by choosing\footnote{For notational convenience the samples are indexed as $c_0,c_1,\ldots$ in this example rather than $c_1,c_2,\ldots$.}
\begin{align}\label{eq:Sexponentials}
s_n(t) =
\begin{cases}
1 & n=0, \\
\cos(2\pi n t/\tau) & 1\leq n\leq M, \\
\sin(2\pi n t/\tau) & M+1\leq n\leq 2M.
\end{cases}
\end{align}
Due to the very small over-sampling factor, only the annihilating filter method is applicable within the approach of \cite{gedalyahu2011}.

Our approach can operate with a budget of only $2M$ samples and with arbitrary sampling kernels, as long as \eqref{eq:s_stability} is satisfied. Nevertheless, we now wish to demonstrate that our method is advantageous over that of \cite{gedalyahu2011} even in settings in which the sampling kernels are chosen as \eqref{eq:Sexponentials}.


We note that the convergence guarantees we provided in previous sections do not hold when sampling above the rate of innovation. However, in practice, the algorithm performs well also in mild over-sampling scenarios, such as the one treated here.

To compare between iterative recovery and the algorithm of \cite{gedalyahu2011}, we concentrated on signals with period $\tau=1$ comprising $M=2$ pulses and thus used $N=2M+1=5$ samples to recover them. We chose a pulse with Fourier coefficients $\tilde{g}_k=1/(5+n^2)$, which, as shown in Fig.~\ref{fig:MSEvsSNR_TomerKifr_OneOverN}\subref{fig:PulsesShape_OneOverN} is very wide in the time domain. This renders the determination of pulse positions a challenging task. The constraints \eqref{eq:ConstrTD} were the same as in Section~\ref{sec:GaussiansNonlin}. The true time delays were $t_1 = 1/\sqrt{15}\approx0.2582$ and $t_2 = 1/\sqrt{2}\approx 0.7071$ and the true amplitudes were randomly generated to yield $a_1\approx0.5285$ and $a_2\approx0.14$. The initialization of the algorithm was the same as in Section~\ref{sec:GaussiansNonlin}. At each iteration of the algorithm, the matrices $\bA$ and $\bB$ comprise the (weighted) Fourier coefficients of shifted versions of $g(t)$ and of $g'(t)$. These quantities can be obtained analytically from the Fourier coefficients of $g(t)$.

In Fig.~\ref{fig:MSEvsSNR_TomerKifr_OneOverN}\subref{fig:MSEvsSNR_TomerKifr} the performance of both approaches is compared against the CRB when the samples are contaminated by white Gaussian noise. As can be seen, the quasi-Newton method outperforms the annihilating-filter-based algorithm at all SNR.

\begin{figure}[t]\centering
\subfloat[]{\includegraphics[scale=0.4, trim=0cm 0.1cm 0cm 0.1cm, clip] {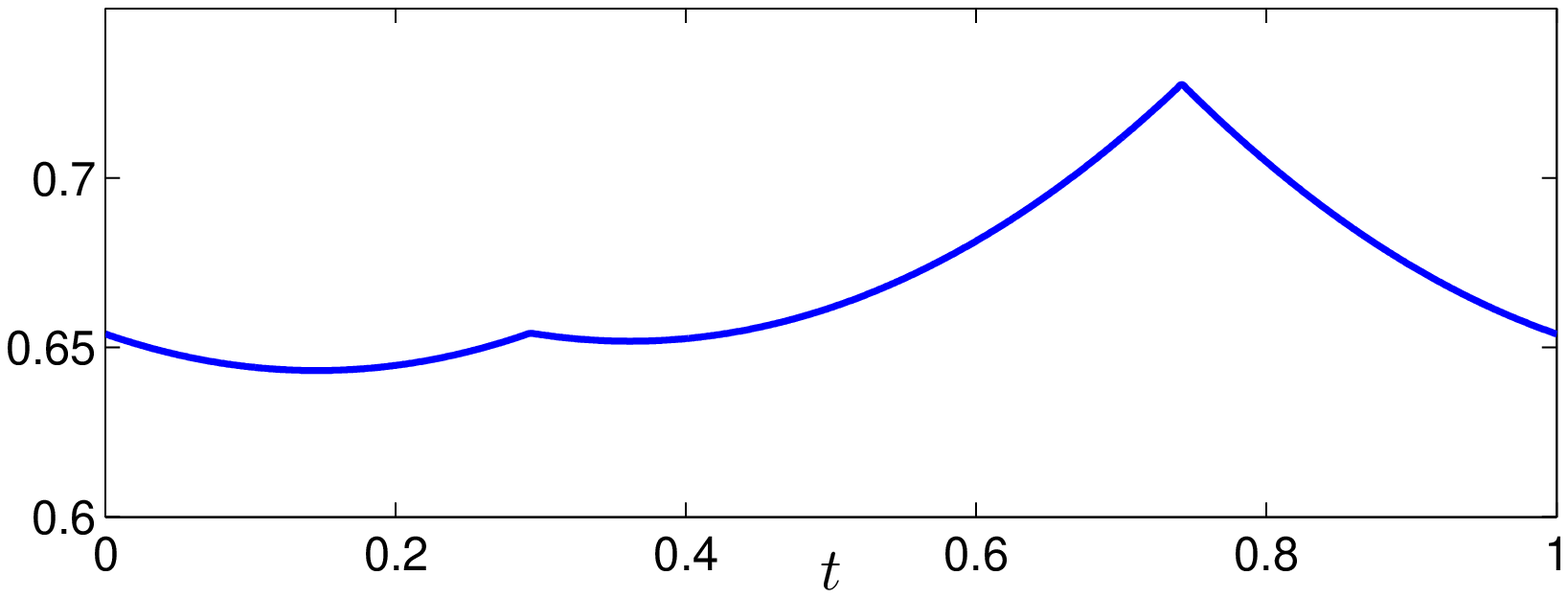}\label{fig:PulsesShape_OneOverN}}\\
\subfloat[]{\includegraphics[scale=0.4, trim=0cm 0.1cm 0cm 0.1cm, clip] {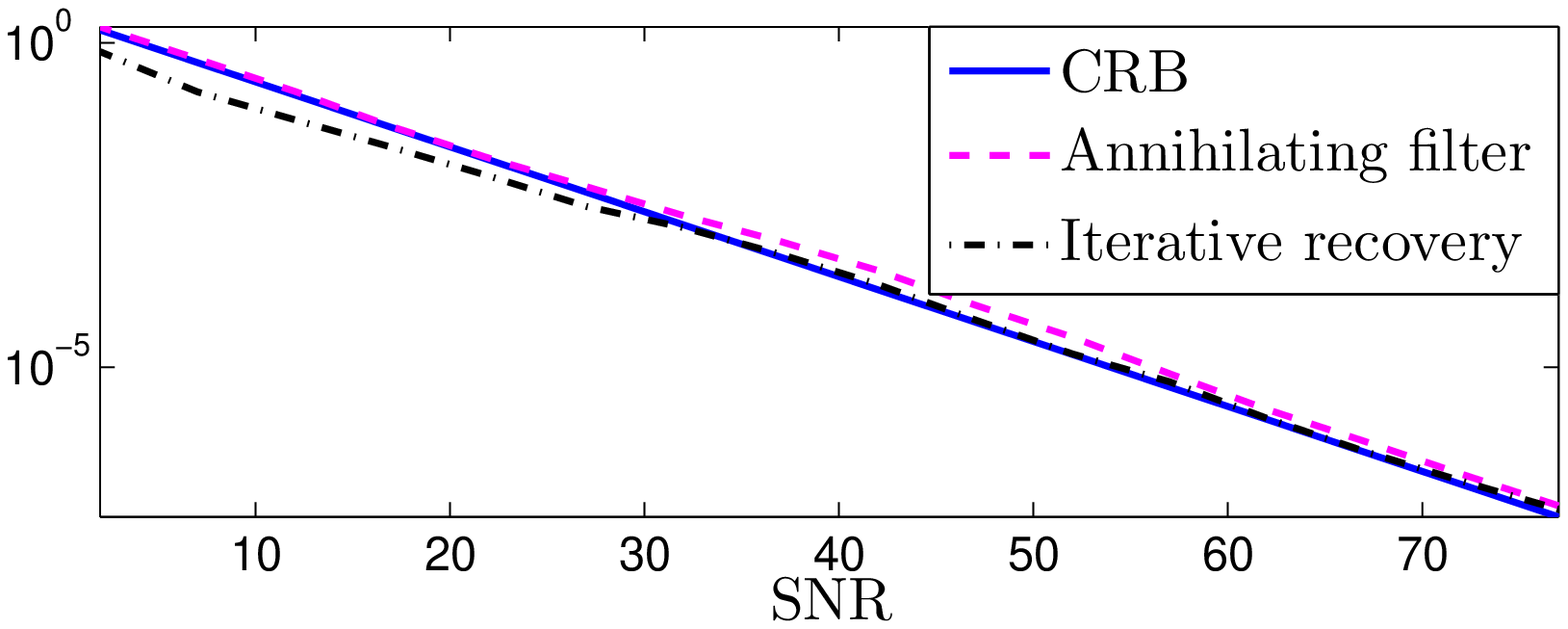}\label{fig:MSEvsSNR_TomerKifr}}
\caption{(a)~One period of the periodic signal $x(t)$ comprising two wide pulses. (b)~MSE as a function of SNR for recovery with the $N=5$ samples corresponding to \eqref{eq:Sexponentials}. The dashed and dash-dotted lines correspond, respectively, to the method of \cite{gedalyahu2011} and quasi-Newton iterations. The solid line corresponds to the CRB for estimating $x(t)$ from the samples.}
\label{fig:MSEvsSNR_TomerKifr_OneOverN}
\end{figure}

\subsection{Stability}
\label{sec:stability}
Although time-delay estimation is a long-studied problem, stability was not given much attention in past works. In \cite{blumensath2011} an example was presented in which the delay $t_1$ of a rectangular pulse $g(t-t_1)$ can be determined from uniformly-spaced samples taken at the output of a triangular impulse-response filter, but this cannot be achieved in a stable manner. For a general channel-sounding setting, it is not trivial to obtain simple-to-verify conditions on the pulse shape $g(t)$, sampling functions $\{s_n(t)\}$, nonlinearity $f(\cdot)$ and the parameters $T_{\min}$, $T_{\max}$, $t_0$ and $a_0$ such that stable recovery is guaranteed. However, as we now demonstrate, unstable settings can be identified numerically using the proposed approach.

Assume that $N=2M$ samples are obtained with a monotonic nonlinearity $f(\cdot)$ and a set $\{s_n\}_{n=1}^{2M}$ of linearly independent sampling kernels. In this case, condition \eqref{eq:empthyNulspace} is satisfied. Assume further that for a certain parameter vector $\bth_0=\begin{pmatrix}t_1&\cdots&t_M&a_1&\cdots&a_M\end{pmatrix}$ and certain initial guess $\bth^0$, the algorithm terminates at a point $\bth_1$ for which $\varepsilon(\bth_1)\neq 0$. This means that $\bth_0$ is not the unique stationary point of the LS objective so that, according to Theorem~\ref{thm:UniqeStationaryPoint}, stable recovery is not possible in this setting for all $\bth$ in the constraint set $\AAA$. More specifically, either condition \eqref{eq:h_stability} or \eqref{eq:s_stability} (or both) are violated for some $\bth\in\AAA$.

In fact, the point at which \eqref{eq:h_stability} or \eqref{eq:s_stability} are violated is no other than $\bth_1$. Indeed, the fact that $\nabla\varepsilon(\bth_1)=0$ and $\varepsilon(\bth_1)\neq0$ implies that $(\partial\hat{\bc}/\partial\bth)|_{\bth_1}=0$ (see \eqref{eq:gradC} and \eqref{eq:ls}). Therefore, by the definition of the Fr\'{e}chet derivative,
\begin{align}
0 &= \lim_{\bdelta\rightarrow 0} \frac{\left\|\hat{\bc}(\bth_1+\bdelta)-\hat{\bc}(\bth_1)\right\|_{\RN^{2M}}}{\|\bdelta\|_{\RN^{2M}}} \nonumber\\
&= \lim_{\bdelta\rightarrow 0} \frac{\left\|S(h(\bth_1+\bdelta))-S(h(\bth_1))\right\|_{\RN^{2M}}}{\|\bdelta\|_{\RN^{2M}}},
\end{align}
contradicting the requirements \eqref{eq:h_stability} and \eqref{eq:s_stability} that
\begin{align}
\left\|S(h(\bth_1+\bdelta))-S(h(\bth_1))\right\|_{\RN^{2M}}\geq\alpha_s\alpha_h\|\bdelta\|_{\RN^{2M}}.
\end{align}
This can also be seen from an estimation viewpoint. Namely, suppose that the samples $\bc$ are perturbed by white Gaussian noise with variance $\sigma^2$. Then the unbiased CRB for estimating $\bth$ from these noisy measurements is given at $\bth=\bth_1$ by \cite{ben-haim2011}
\begin{equation}
\sigma^2\left(\left(\left.\pd{\hat{\bc}}{\bth}\right|_{\bth_1}\right)^*\left(\left.\pd{\hat{\bc}}{\bth}\right|_{\bth_1}\right)\right)^{-1}.
\end{equation}
If $(\partial\hat{\bc}/\partial\bth)|_{\bth_1}=0$ then there exists no unbiased technique that can recover the parameters with a finite MSE.

As a demonstration of the utilization of this approach, consider again the setting of Section~\ref{sec:KfirRonen}. As mentioned above, existing techniques that do not involve discretization can only handle the case in which the frequencies of the sampling kernels are consecutive. An interesting question is whether there is a potential gain in using non-consecutive indices. To study this setting, we used our algorithm to recover two time delays, where $g(t)$ was taken to be a pulse whose Fourier coefficients are equal~$1$ up to some large index and~$0$ otherwise. We used four sinusoidal sampling functions (two sines and two cosines) with frequencies $1$ and $3$. While the true parameters were $\begin{pmatrix}t_1 & t_2 & a_1 & a_2\end{pmatrix}=\begin{pmatrix}0.2 & 0.8 & 1 & 5\end{pmatrix}$, the algorithm converged to the point $\begin{pmatrix}0.34 & 0.85 & 0.41 & 3.1\end{pmatrix}$. This means that the CRB for estimating $\bth$ explodes at this point. Figure~\ref{fig:CRBvst2} depicts the CRB as function of $t_2\in[0.85,1]$ for $t_1=0.34$, verifying that this is indeed the case. We therefore conclude that in this setting there exist parameter values that \emph{cannot be recovered stably by any technique}.

\begin{figure}[t]\centering
\includegraphics[scale=0.4]{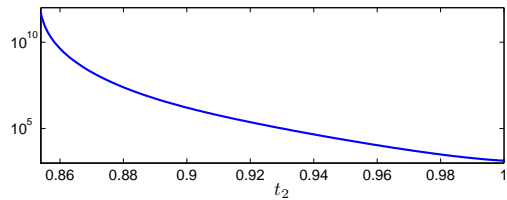}
\caption{CRB versus $t_2$ for fixed $t_1$ in a setting with sinusoidal sampling kernels with nonconsecutive frequencies.}\label{fig:CRBvst2}
\end{figure}

A word of caution is in place, though. For our approach to be able to recover a parameter vector $\bth_0$, we need that \emph{every} $\bth\in\AAA$ can be stably recovered and not only $\bth_0$ itself. Therefore, the fact that in some settings with nonconsecutive sampling frequencies there exist unstable points in $\AAA$ limits the applicability of our method in those scenarios. It may thus be of interest in certain applications to pursue methods that can recover any stably-reconstructible $\bth_0$, regardless if there exist other points in $\AAA$ at which the CRB is infinite.

\section{Application to CPM Communication}
\label{sec:ApplicationFSK}

As mentioned in Section~\ref{sec:ProblemSetting}, an important application area not treated in the FRI literature is CPM communication (see \eqref{eq:CPM}). For a general rational modulation index $h$ and pulse width $L$, optimum coherent detection can be performed by means of the Viterbi algorithm. A major limitation with this approach, though, is that it requires sampling at a rate of $1/T$ at the output of $4Q^L$ linear filters \cite{sundberg86}. This corresponds to an over-sampling factor of $4Q^L$ beyond the rate of innovation. Furthermore, for $h=2k/p$, where $k$ and $p$ have no common factors, the number of states in the Viterbi decoder is $pQ^{L-1}$. Here we propose a sub-optimal alternative, which employs an average sampling rate of only $2/T$, as depicted in Fig.~\ref{fig:FSKreceiver}. Our approach consists in treating the data symbols $\{a_m\}$ in \eqref{eq:CPM} as continuous-valued and quantizing the resulting recoveries to the nearest element in the set $\{\pm1,\ldots,\pm(Q-1)\}$. We emphasize that our proposed approach does not perform well in noise and serves here merely as a demonstration of treatment of non union-of-subspace models.

\begin{figure}[t]\centering
\includegraphics[scale=0.85]{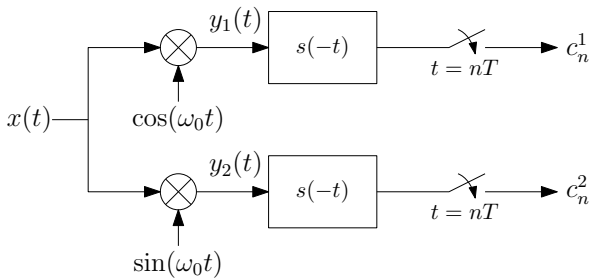}
\caption{Proposed CPM receiver.}
\label{fig:FSKreceiver}
\end{figure}

In principle, cleverly designed measurements at a rate of $1/T$ should suffice (in the noiseless setting) for perfect recovery. However, as we will see, neither of the branches of Fig.~\ref{fig:FSKreceiver} suffices by itself for recovery of all symbols with our iterative approach. Instead, we will alternately use bunches of samples from each of the branches. The signals $y_1(t)$ and $y_2(t)$ contain one replica of the frequency content of $x(t)$ around $\omega=0$ and one around $2\omega_0$. Suppose for the moment that the filter $s(-t)$ suppresses the replica around $2\omega_0$ so that, to high precision, for $i=1,2$,
\begin{align}
y_i(t) &= f_i\left(\sum_{m}\tilde{a}_m \tilde{g}(t-mT)\right),
\end{align}
where we adopted the representation \eqref{eq:FSK} and denoted $f_1(\alpha)=0.5\cos(\alpha)$ and $f_2(\alpha)=-0.5\sin(\alpha)$.
Thus, for $i=1,2$,
\begin{equation}\label{eq:sampCPM}
c^{i}_{n} = \int_{-\infty}^\infty s(t-nT) f_i\left(\sum_{m\in\ZN}\tilde{a}_m\tilde{g}(t-nT)\right) dt.
\end{equation}

Linear sampling of a SI signal passing through memoryless nonlinearity, as in \eqref{eq:sampCPM}, was studied in \cite{dvorkind2008,faktor10}. In particular, it was shown that if the nonlinearity is a monotone function that does not vary too rapidly, then a stationary point of the LS objective is necessarily a global minimum. In our setting, neither $f_1(\alpha)$ nor $f_2(\alpha)$ are monotone functions. However, since $\tilde{a}_m$ can only vary by $\pm1$ from one symbol to the next, the phase
\begin{equation}
\varphi(t) = \sum_{m}\tilde{a}_m \tilde{g}(t-mT).
\end{equation}
is guaranteed to vary by less than $\pm\pi/2$ over short enough time segments. Specifically, $f_i(\varphi(t))$ is a monotone function of $\varphi(t)$ over a certain time interval if
\begin{equation}\label{eq:phaseBounds1}
(i-1)\pi/2+ 2\pi p<\varphi(t)<(i+1)\pi/2+ 2\pi p
\end{equation}
or
\begin{equation}\label{eq:phaseBounds2}
(i+1)\pi/2+ 2\pi p<\varphi(t)<(i+3)\pi/2+ 2\pi p
\end{equation}
for some $p\in\ZN$ throughout this period. For such a segment $[t_1,t_2]$ and assuming that the support of $s(t)$ is contained in $[t_a,t_b]$, all samples $c^i_n$ with indices $(t_1-t_a)/T<n<(t_2-t_b)/T$ conform to the model in \cite{dvorkind2008,faktor10}. These samples can be used to recover a corresponding set of symbols.

To summarize, our approach for the simple setting in which $s(t)$ is supported on $[0,T]$ is as follows. Suppose that all symbols up to index $n_1$ were recovered. These allow to determine $\varphi(n_1T)$, which is used to decide, according to \eqref{eq:phaseBounds1} and \eqref{eq:phaseBounds2}, weather the next batch of samples is to be taken from the first branch or from the second one. Next, the maximal index $n_{\max}$ such that the phase remains within the corresponding interval for every $t\in[(n_1+1)T,n_{\max}T]$ is determined\footnote{This can be done by noting that the change in phase for $t\geq(n_1+1)T$ is due both to the contribution of the known symbols $\{a_m\}_{m\leq n_1}$ and to the symbols $\{a_m\}_{m> n_1}$, which are yet to be recovered. The largest change occurs if the latter are all $+1$ or $-1$.}. The samples with indices $n_1+1,\ldots,n_{\max}-1$ are then used to recover the symbols with the corresponding indices. This process is repeated sequentially.

The $n$th sample in the $i$th channel is given by $c^{i}_n = \langle y_i, s_n \rangle$, where $s_n(t) = s(t-nT)$. Assume, without loss of generality, that $\bth=\begin{pmatrix}a_1,\ldots,a_M\end{pmatrix}$.
Direct computation shows that
\begin{align}
\pd{\hat{\bc}^i}{\bth} =
\begin{pmatrix}
\langle z^i_1, s_1\rangle & \cdots & \langle z^i_M, s_1\rangle \\
\vdots & & \vdots\\
\langle z^i_1, s_N\rangle & \cdots & \langle z^i_M, s_N\rangle
\end{pmatrix},
\end{align}
where
\begin{align}
z^1_m(t) &= \frac{1}{2}q(t-mT)\left(\cos\left(2\omega_0 t + \varphi(t)\right)+\cos\left(\varphi(t)\right)\right),\nonumber\\
z^2_m(t) &= \frac{1}{2}q(t-mT)\left(\sin\left(2\omega_0 t + \varphi(t)\right)- \sin\left(\varphi(t)\right)\right),
\end{align}
and we denoted $q(t) = \int_{-\infty}^t g(\tau)d\tau$. To account for the fact that $|a_m|\leq Q-1$, we chose to enforce the constraint $|a_m|<Q$ by using the parametrization $\tilde{\theta}_m = \tan(\pi a_m/(2Q))$.
The derivative of the corresponding transformation $\bth=p(\tilde{\bth})$ is $\partial p/ \partial\tilde{\bth} = (2Q/\pi){\rm diag}
\begin{pmatrix}
1/(1+\tilde{\theta}_1^2) & \cdots & 1/(1+\tilde{\theta}_M^2)
\end{pmatrix}$.

Figure~\ref{fig:CPMmodulation} shows the phase of a typical binary CPM signal (namely, with $Q=1$) with modulation index $h=1/7$ and with the 5REC pulse $g(t)={\rm rect}_{[0,5T]}(t)$. Figure~\ref{fig:CPMmodulation}\subref{fig:symbolsRec} shows the recovery of the symbols with only $2$ iterations per batch of samples. Here the sampling kernels were taken as $s(t) = {\rm rect}_{[0,T]}(t)$. The batches of samples on which the algorithm operated are marked with dashed vertical lines. As can be seen, even with two iterations, the original symbols can be recovered by quantization of the recovered symbols.

\begin{figure}[t]\centering
\subfloat[]{\includegraphics[scale=0.4, trim=0cm 0.1cm 0cm 0.1cm, clip]{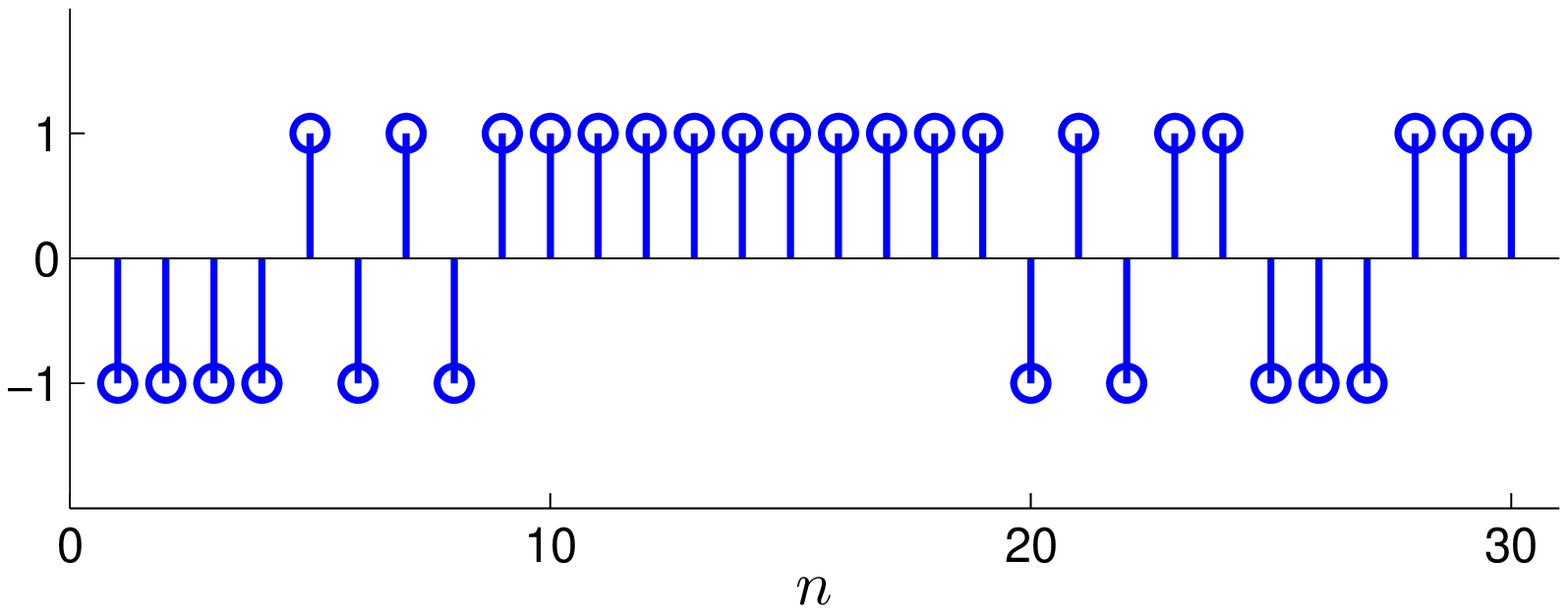}\label{fig:symbols}}\\
\subfloat[]{\includegraphics[scale=0.4, trim=0cm 0.1cm 0cm 0.1cm, clip]{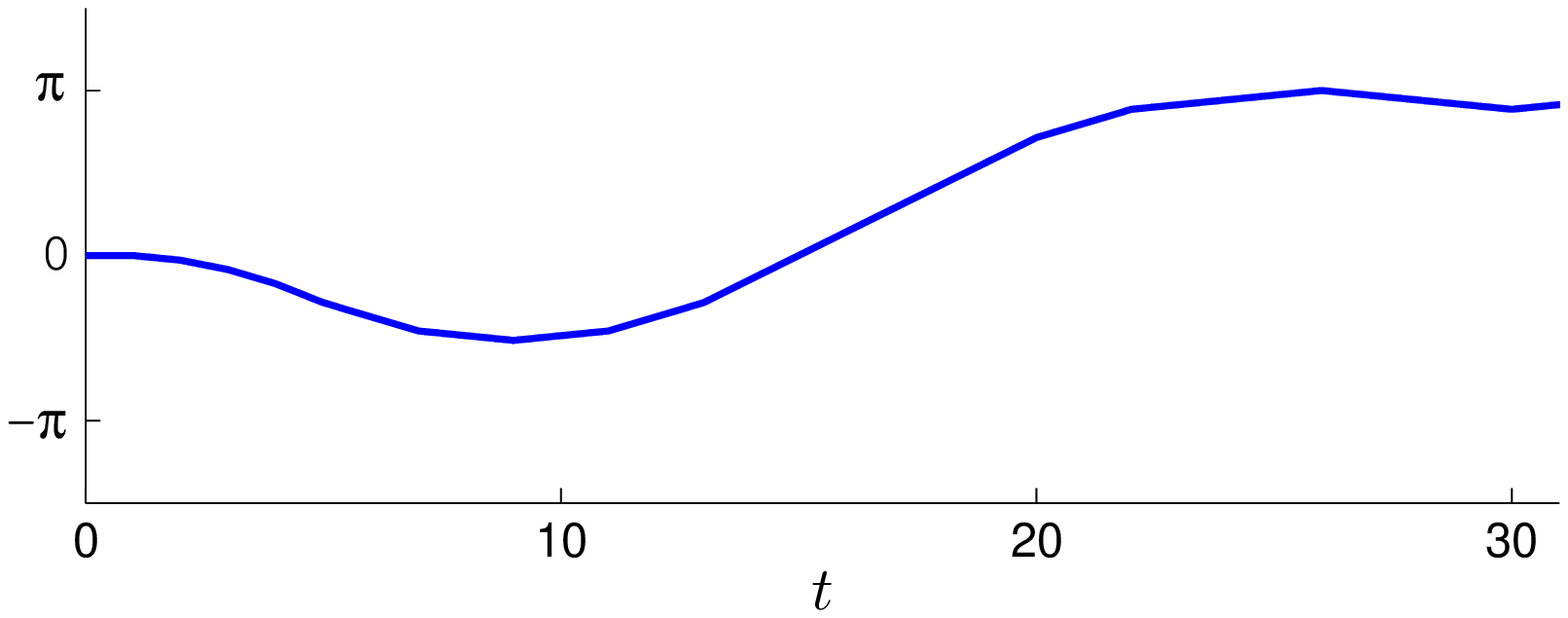}\label{fig:phase}}\\
\subfloat[]{\includegraphics[scale=0.4, trim=0cm 0.1cm 0cm 0.1cm, clip]{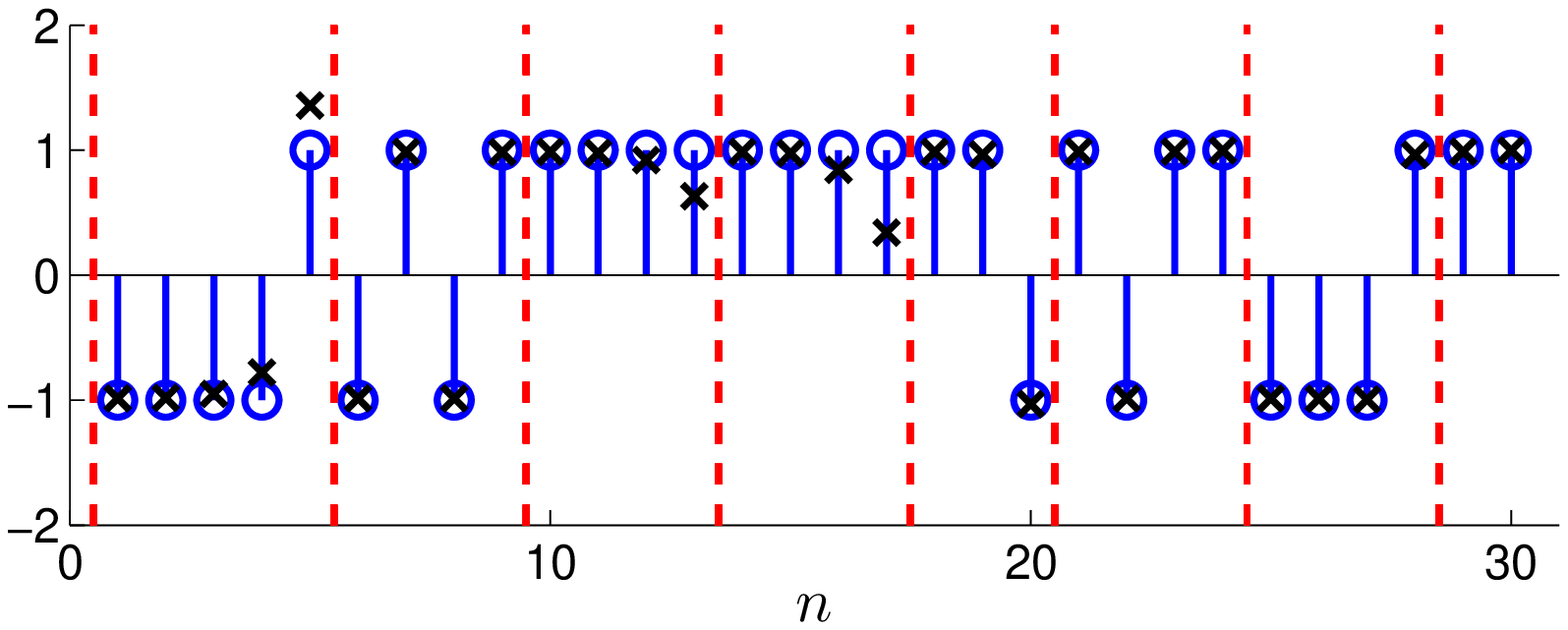}\label{fig:symbolsRec}}
\caption{Binary 5REC CPM modulation with index $h=1/7$. (a)~Symbols $a_m$. (b)~Corresponding phase $\varphi(t)$. (c)~The 'x'-marks denote the recovered coefficients using $2$ quasi-Newton iterations (before quantization).}
\label{fig:CPMmodulation}
\end{figure}

\section{Conclusion}

In this paper, we studied recovery of the parameters defining an FRI signal from samples taken at the rate of innovation. We showed that in any situation in which the parameters can be stably recovered, this can be achieved by a general-purpose unconstrained optimization method. Our approach thus provides a simple means for treating a wide range of FRI signal classes and sampling methods. We demonstrated the usefulness of our strategy in reconstructing finite and periodic pulse streams from nonlinear and nonideal samples as well as in decoding CPM modulated messages. We also showed that our method is often advantageous in noisy settings.

\appendices
\section{Convergence of Quasi-Newton Iterations}
\label{sec:ConvNewton}
Letting $\bQ=\partial{\hat{\bc}}/\partial{\bth}|_{\bth^\ell}$ and substituting $\bB^\ell=(\bQ^*\bQ)^{-1}$ and \eqref{eq:gradC}, the left-hand side of \eqref{eq:cos} becomes
\begin{align}
&\frac{(\hat{\bc}(\bth^\ell)-\bc)^*\bQ(\bQ^*\bQ)^{-1}\bQ^*(\hat{\bc}(\bth^\ell)-\bc)} {\left\|(\bQ^*\bQ)^{-1}\bQ^*(\hat{\bc}(\bth^\ell)-\bc)\right\|\left\|\bQ^*(\hat{\bc}(\bth^\ell)-\bc)\right\|} \nonumber\\
&\hspace{2cm}=\frac{\left\|\hat{\bc}(\bth^\ell)-\bc\right\|^2} {\left\|\bQ^{-1}(\hat{\bc}(\bth^\ell)-\bc)\right\|\left\|\bQ^*(\hat{\bc}(\bth^\ell)-\bc)\right\|} \nonumber\\
&\hspace{2cm}\geq\frac{1}{\left\|\bQ^{-1}\right\|\left\|\bQ\right\|}.
\end{align}
Here, we used the fact that $\Ra{\bQ}=\RN^K$, which was established in the proof of Theorem~\ref{thm:UniqeStationaryPoint}, so that $\bQ(\bQ^*\bQ)^{-1}\bQ^*=\bI$. Now, the right-hand side of \eqref{eq:s_stability}, together with \eqref{eq:beta_h}, imply that $\|\bQ\|\leq\beta_s\beta_h$. Similarly, the left-hand side of \eqref{eq:s_stability}, together with \eqref{eq:h_stability}, imply that $\|\bQ^{-1}\|\leq1/(\alpha_s\alpha_h)$. Therefore,
\begin{equation}
\frac{1}{\left\|\bQ^{-1}\right\|\left\|\bQ\right\|} \geq \frac{\beta_s\beta_h}{\alpha_h\alpha_s},
\end{equation}
so that \eqref{eq:cos} is satisfied with any $\delta<(\beta_s\beta_h)/(\alpha_h\alpha_s)$.

\section{Proof of Gradient Lipschitz Continuity}
\label{sec:GradLip}
Denoting $\be(\bth)=\hat{\bc}(\bth)-\bc$, we have
\begin{align}\label{eq:proogGradLip1}
&\|\nabla\varepsilon(\bth_1)-\nabla\varepsilon(\bth_2)\| = \nonumber\\
&\hspace{0.7cm}=\left\|\left(\left.\pd{\hat{\bc}}{\bth}\right|_{\bth_1}\right)^* \be(\bth_1) - \left(\left.\pd{\hat{\bc}}{\bth}\right|_{\bth_2}\right)^* \be(\bth_2)\right\| \nonumber\\
&\hspace{0.7cm}=\frac{1}{2}\left\|
\left(\left.\pd{\hat{\bc}}{\bth}\right|_{\bth_1}-\left.\pd{\hat{\bc}}{\bth}\right|_{\bth_2}\right)^*(\be(\bth_1)+\be(\bth_2))\right.\nonumber\\
&\hspace{1.7cm}\left.+\left(\left.\pd{\hat{\bc}}{\bth}\right|_{\bth_1}+\left.\pd{\hat{\bc}}{\bth}\right|_{\bth_2}\right)^*(\be(\bth_1)-\be(\bth_2))\right\|
\nonumber\\
&\hspace{0.7cm}\leq\frac{1}{2}\left\|
\left.\pd{\hat{\bc}}{\bth}\right|_{\bth_1}-\left.\pd{\hat{\bc}}{\bth}\right|_{\bth_2}\right\|\|\be(\bth_1)+\be(\bth_2)\|\nonumber\\
&\hspace{1.7cm}+\frac{1}{2}\left\|
\left.\pd{\hat{\bc}}{\bth}\right|_{\bth_1}+\left.\pd{\hat{\bc}}{\bth}\right|_{\bth_2}\right\|\|\be(\bth_1)-\be(\bth_2)\|.
\end{align}
Assuming that $\bth_1,\bth_2\in\mathcal{N}$, conditions \eqref{eq:s_stability} and \eqref{eq:hprime_stability} imply that
\begin{align}\label{eq:proogGradLip2}
\left\|\left.\pd{\hat{\bc}}{\bth}\right|_{\bth_1}-\left.\pd{\hat{\bc}}{\bth}\right|_{\bth_2}\right\| &\leq
\beta_s \left\|\left.\pd{h}{\bth}\right|_{\bth_1}-\left.\pd{h}{\bth}\right|_{\bth_2}\right\| \nonumber\\
&\leq \beta_s\beta_{h'}\|\bth_1-\bth_2\|. 
\end{align}
Furthermore, \eqref{eq:s_stability} implies that
\begin{align}\label{eq:proogGradLip3}
\left\|\left.\pd{\hat{\bc}}{\bth}\right|_{\bth_1}+\left.\pd{\hat{\bc}}{\bth}\right|_{\bth_2}\right\| \leq
\left\|\left.\pd{\hat{\bc}}{\bth}\right|_{\bth_1}\right\|+\left\|\left.\pd{\hat{\bc}}{\bth}\right|_{\bth_2}\right\| \leq 2\beta_s. \end{align}
Since $\bth_1,\bth_2\in\mathcal{N}$, it also follows that
\begin{equation}\label{eq:proogGradLip4}
\|\be(\bth_1)+\be(\bth_2)\|\leq\|\be(\bth_1)\|+\|\be(\bth_2)\|\leq 2\varepsilon(\bth^0).
\end{equation}
Finally,
\begin{equation}\label{eq:proogGradLip5}
\|\be(\bth_1)-\be(\bth_2)\|=\|\hat{\bc}(\bth_1)-\hat{\bc}(\bth_2)\|\leq \beta_s\|\bth_1-\bth_2\|.
\end{equation}
Substituting \eqref{eq:proogGradLip2}, \eqref{eq:proogGradLip3}, \eqref{eq:proogGradLip4} and \eqref{eq:proogGradLip5} into \eqref{eq:proogGradLip1} yields
\begin{align}
\|\nabla\varepsilon(\bth_1)-\nabla\varepsilon(\bth_2)\| \leq \left(\beta_s\beta_{h'}\varepsilon(\bth^0)+\beta_s^2\right)\|\bth_1-\bth_2\|,
\end{align}
which proves that $\nabla\varepsilon(\bth_1)$ is Lipschitz continuous over $\mathcal{N}$ with Lipschitz bound $\beta_s(\beta_{h'}\varepsilon(\bth^0)+\beta_s)$.

\bibliographystyle{IEEEtran}

\end{document}